\documentclass[10pt,journal]{IEEEtran} 
\usepackage[utf8]{inputenc}
\usepackage[T1]{fontenc}
\usepackage{amssymb,amsmath}
\usepackage{cite}
\usepackage{psfrag}
\usepackage{url}
\usepackage[absolute,overlay]{textpos}
\DeclareMathOperator{\Tr}{Tr}
\usepackage{pgf}
\usepackage[export]{adjustbox}
\usepackage{bm}
\usepackage{bbm}
\usepackage{booktabs}
\usepackage{url}
\usepackage{upgreek}
\usepackage{verbatimbox}
\usepackage{algorithm}
\usepackage{algorithmic}
\usepackage{enumitem}
\usepackage{caption}
\usepackage{amsthm}
\usepackage{graphicx}
\usepackage{subfig}

\newtheorem{lemma}{Lemma}
\newtheorem{theorem}{Theorem}

\newtheorem{remark}{Remark}

\newcommand{\appropto}{\mathrel{\vcenter{
\offinterlineskip\halign{\hfil$##$\cr		\propto\cr\noalign{\kern2pt}\sim\cr\noalign{\kern-2pt}}}}}
\usepackage{cancel}
\usepackage{colortbl}

\begin{document}	
	\title{``Iridescent'' Reflective Tags to Enable Radar-based Orientation Estimation}		
	\author{Onel L. A. López,~\IEEEmembership{Senior Member,~IEEE,} Zhu Han,~\IEEEmembership{Fellow,~IEEE}, and Ashutosh Sabharwal,~\IEEEmembership{Fellow,~IEEE}
    \thanks{O. López is with the Centre for Wireless Communications, University of Oulu, Finland, e-mail: onel.alcarazlopez@oulu.fi. Z. Han is with the Department of Electrical and Computer Engineering, University of Houston, Houston, Texas, USA, email: zhan2@central.uh.edu. A. Sabharwal is with the Department of Electrical and Computer Engineering, Rice University, Houston, Texas, USA, e-mail: ashu@rice.edu.}
    \thanks{The work of O. L\'opez was partially supported by the Research Council of Finland (former Academy of Finland) 6G Flagship Programme (Grant Number: 369116) and ECO-LITE (Grant NUmber: 362782), the KAUTE Foundation (\emph{Tutkijat maailmalle} program), the Nokia Foundation (\emph{Jorma Ollila} grant), the European Commission through the Horizon Europe/JU SNS project Hexa-X-II (Grant Agreement no. 101095759), and the Finnish-American Research \& Innovation Accelerator. The work of Z. Han was partially supported by NSF  CNS-2107216, CNS-2128368, CMMI-2222810, and ECCS-2302469. The work of A. Sabharwal was partially supported by NSF Grants 1956297 and 2215082.}
 } 	
\maketitle
\begin{abstract}
    Accurate orientation estimation of objects can aid in scene understanding in many applications. In this paper, we consider use cases where passive tags could be deployed to assist radar systems in estimating object orientation. Towards that end, we propose the concept of passive iridescent reflective tags that selectively reflect different wavelengths in different directions. We propose a conceptual tag design based on leaky-wave antennas. We develop a framework for signal modeling and orientation estimation with a multi-tone radar. We analyze the impact of imperfect tag location information, revealing that it minimally impacts orientation estimation accuracy. To reduce estimator complexity, we propose a radiation pointing angle-based estimator with near-optimal performance. We derive its feasible orientation estimation region and show that it depends mainly on the system bandwidth. Monte Carlo simulations validate our analytical results while evincing that the low-complexity estimator achieves near-optimal accuracy and that its feasible orientation estimation region closely matches that of the other estimators. Finally, we show that the optimal number of tones increases with the sensing time under a power budget constraint, multipath effects may be negligible, signal-to-noise ratio gains rise with the number of tones, and many radar antennas can hurt estimation performance when the signal contains very few tones.   
\end{abstract}
\IEEEpeerreviewmaketitle
\section{Introduction}\label{intro}
Radar-based sensing is prevalent in many applications, and is especially effective in visually degraded environments, e.g., affected by fog or smoke, where other imaging modalities are rendered ineffective. However, at millimeter (mm)-wave bands, many real-life objects are specular; thus, scene reconstruction depends on their orientation with respect to the radar. Consider the examples in Fig.~\ref{fig:examples} where the orientation of the door or package determines the strength of the return to the radar. Weak returns hinder accurate scene reconstruction by the radar, leading to poor scene understanding and impaired subsequent tasks. In this paper, we focus on radar-based object orientation estimation to improve scene understanding, which is  important for efficient navigation, manipulation, interaction, alignment, and 3D vision in fields like robotics.
\begin{figure}[t!]
    \centering
    \includegraphics[width=\linewidth]{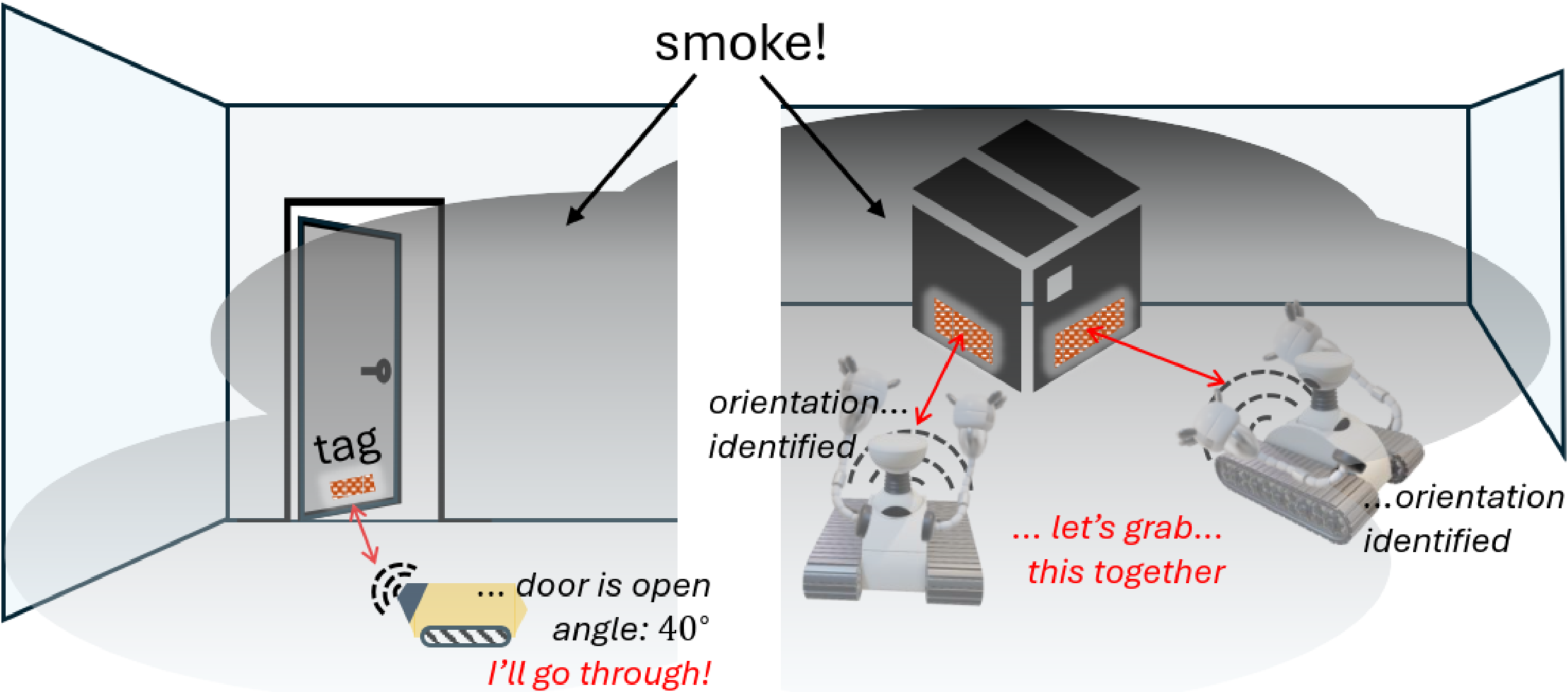}
    \caption{Two example use cases where radar-equipped mobile robots are operating in visually challenging conditions: (left) the robot aims to locate the door to a room and understand if the door is open or closed, and (right) two robots try to locate an object that needs to be retrieved. In both cases, understanding the orientation of objects in the scene, the door or package, can aid the robots in completing their task.}
   \label{fig:examples}
\end{figure}

A way to address the above challenge is to engineer the objects to improve their radiofrequency (RF)-detectability/recognition, e.g., using advanced manufacturing (meta-)materials \cite{Suresh.2021,Padilla.2022,Xue.2023}, attaching backscattering tags \cite{Shirehjini.2012,Wei.2016,Figueiredo.2023,Krigslund.2012,Genovesi.2018,Barbot.2020,Liu.2024,Chang.2020,Rammal.2023,Xue.2023}, and/or applying surface modifications \cite{Xue.2023} that interact with RF signals in predictable ways to passively (and indirectly) indicate object's physical properties. This is akin to ``painting an object'' to be properly noticed by a camera \cite{Chang.2020,Rammal.2023}. From these approaches, using passive RF identification (RFID) tags as orientation sensors has been proposed in past work~\cite{Shirehjini.2012,Wei.2016,Figueiredo.2023,Krigslund.2012,Genovesi.2018,Barbot.2020,Liu.2024,Chang.2020,Rammal.2023}. 

The proposals in \cite{Shirehjini.2012,Wei.2016,Figueiredo.2023} employ conventional RFID systems with integrated circuits where tags encode data like identification and relative position, while chipless RFID tags are used in \cite{Krigslund.2012,Genovesi.2018,Barbot.2020,Liu.2024,Chang.2020,Rammal.2023} to favor lower complexity. However, i) the tags in \cite{Krigslund.2012,Genovesi.2018,Barbot.2020,Liu.2024,Chang.2020,Rammal.2023} still store/encode some data; ii) dual-polarized antennas are required in \cite{Krigslund.2012,Genovesi.2018,Barbot.2020} and the corresponding polarization diversity or mismatch-based encoding techniques are affected by distance, causing undesired received signal strength variations and limiting the sensing range, and iii) multiple tags are required in \cite{Shirehjini.2012,Wei.2016,Figueiredo.2023,Krigslund.2012,Liu.2024,Chang.2020,Rammal.2023}. All these complicate tag deployment in terms of system solution feasibility and scalability, calling for lower-complexity designs avoiding chips, complex circuits, extensive calibration, and data storing and encoding capabilities. In this work, we seek an alternate tag design concept to mitigate the above challenges.

We propose a new form of passive tags, which we label \emph{iridescent} tags, that have a frequency selectivity in their angular response. Iridescent tags have maximal reflected energy in a sub-band of the received wavefront spectrum. The idea is inspired by Iridescence~\cite{Srinivasarao.1999} where certain surfaces appear to change color as the viewing angle changes (the phenomenon occurs in nature, e.g., soap bubbles, certain butterflies, and bird feathers).

Our main contributions in this paper are as follows:

\subsubsection{System design}
We propose a conceptual leaky wave antenna (LWA)-equipped tag as a potential option for an iridescent reflective tag. While LWAs have been recently used for angle estimation and related tasks in communication systems~\cite{Kludze.2022,Poveda.2021,Martinez.2022,MartinezPoveda.2022,Neophytou.2022,Martinez.2023,Haofan.2023}, they have not been utilized to estimate object orientation using radars. We propose incorporating a low-complexity backscattering circuit into the tag, e.g., as illustrated in Fig.~\ref{fig1}, to enable a radar to estimate the object orientation by estimating the orientation of the tag attached to the object, as shown in the example use cases in Fig.~\ref{fig:examples}. Notably, the proposed design also enables the use of a single tag for orientation estimation, compared to requiring multiple tags as in~\cite{Shirehjini.2012,Wei.2016,Figueiredo.2023,Krigslund.2012,Liu.2024,Chang.2020,Rammal.2023}. We propose the use of multi-tone wideband radars for the tag orientation sensing/estimation. This is compatible with orthogonal frequency division multiplexing (OFDM)-based transmissions inspired by joint communications and sensing design goals for next-generation wireless networks. 

\begin{figure}[t!]
    \centering
    \includegraphics[width=\linewidth]{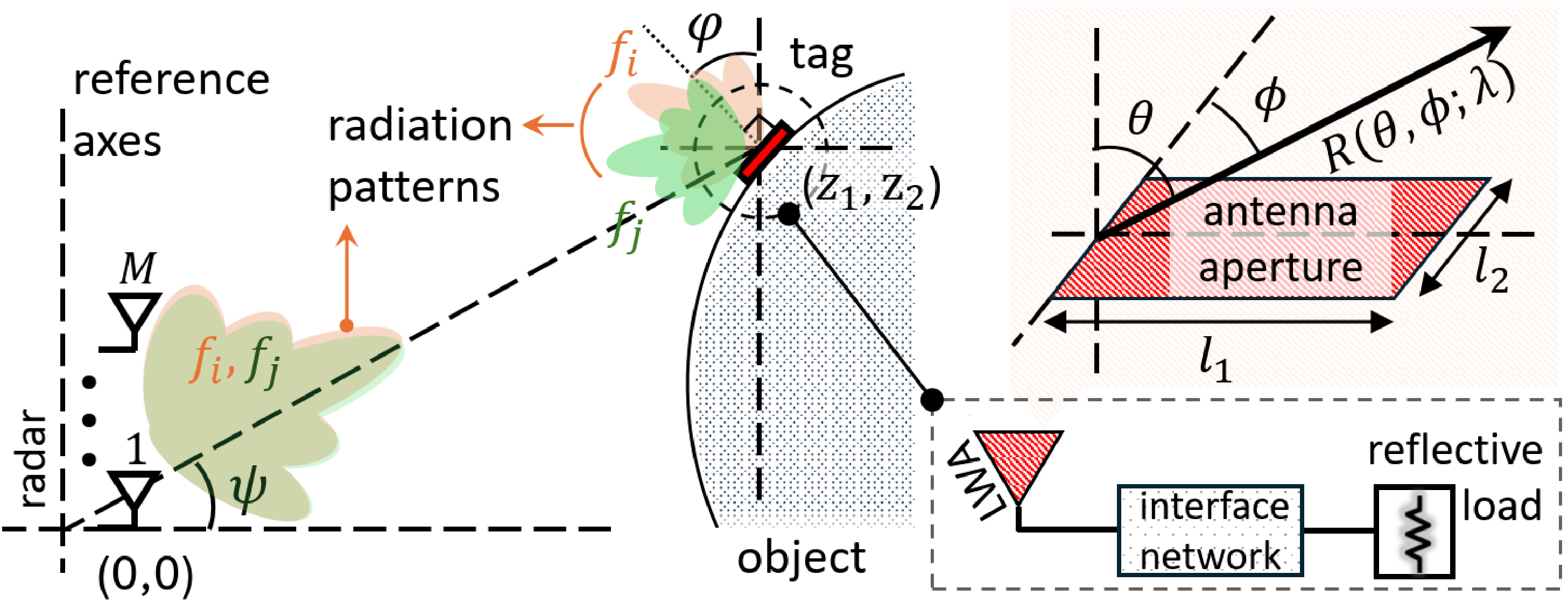}
    \caption{(Left) A radar equipped with $M$ antennas estimates the orientation of iridescent tags, e.g., LWA-equipped passive backscattering tag. (Right) Block diagram depicting the tag block diagram, LWA geometry, and radiation angles.}     
    \label{fig1}
\end{figure}

\subsubsection{Estimator formulation}
    We formulate the maximum likelihood estimator (MLE) for the tag's orientation, and its simplified form under perfect location information, referred to as P$-$MLE. The latter is derived to assess how such an estimator performs even under non-ideal tag location information availability. We formally assess the impact of the tag's imperfect location information on the orientation estimation accuracy. In this regard, we demonstrate that the statistics of the tag's location estimation errors, specifically the estimation error standard deviation, do not influence the orientation estimation accuracy in practical deployments where the tag is relatively far from the radar. We also propose an approximate MLE, referred to as A$-$MLE, which achieves a computational complexity lower than MLE. Finally, we propose an additional simple tag radiation pointing angle-based estimator, referred to as RPA, that achieves near-optimal results. We derive the feasible orientation estimation region of RPA and show that it depends mainly on the system bandwidth. Moreover, we demonstrate that the estimation error asymptotically vanishes as the number of tones and estimation time increase.
\subsubsection{Performance and design insights}
We assess estimators' complexity and radar/tag implementation options, including simultaneous tone and frequency-sweeping transmissions. The complexity of all the estimators is shown to scale with the product of the number of tones and sampling time slots, but the number of radar antennas and brute force sampling points also significantly increase the complexity of A$-$MLE, P$-$MLE, and MLE compared to RPA. MLE even incurs the additional complexity of evaluating an expectation with respect to the tag location uncertainty. We corroborate our analytical insights and assess the performance of the proposed estimators via Monte Carlo simulations. It is shown that the feasible orientation estimation range obtained for RPA also applies to the other estimators. Moreover, A$-$MLE and RPA are shown to perform closely, outperforming P$-$MLE both in estimation accuracy under imperfect tag location information and complexity, and approaching MLE's accuracy with much lower complexity. We show that potential multi-path effects may be negligible and that there is an optimum number of tones given a total power budget constraint. Notably, the multi-path effect vanishes and the optimum number of tones increases with the sensing time,  performance gains from higher signal-to-noise ratio (SNR) increase with the number of tones, and increasing the number of radar antennas deteriorates the system performance when the sensing signals comprise a relatively small number of tones.

The rest of this paper is organized as follows. Section~\ref{system} introduces the system model and tag orientation estimation problem and motivates a wideband design. The latter, together with MLE and P$-$MLE, are presented in Section~\ref{sec3}. Section~\ref{sec4} characterizes the impact of the tag location information accuracy and presents A$-$MLE exploiting related insights. The low-complexity estimator RPA is proposed in Section~\ref{sec5}, and performance trends are provided. Implementation considerations and numerical results are respectively discussed in Sections~\ref{imp} and \ref{numeric}. Section~\ref{conclusions} concludes the article. Table~\ref{tab:symbols} lists the main symbols used in the paper and their default values for performance analysis in Section~\ref{numeric}.
\begin{table}[t!]
    \centering
    \caption{Main symbols utilized throughout the paper and their default values for performance assessment}
    \begin{tabular}{p{0.95cm}p{4.75cm}p{2.05cm}}
     \toprule
       \textbf{symbol}   &  \textbf{meaning}  & \textbf{default value} \\ \midrule
       $M$ & number of radar antennas & 4 \\
        $\varphi$, $\hat{\varphi}$  & true, estimated orientation angle & $-$ \\ 
        $d$ & consecutive radar antenna separation & $3\!\times\! 10^8/(2f_1)$ (m)\\
        $\mathbf{x}_m$ & 2D position of the $m-$th radar antenna & $-$\\
        $\mathbf{z}$ & true 2D position of the tag & $||\mathbf{z}||=20$ m\\
        $\hat{\mathbf{z}},\Delta\mathbf{z}$ & estimate and estimation error of $\mathbf{z}$ & $-$\\
        $\mathbf{\Sigma}$ & estimation error covariance matrix with diagonal entries $\sigma_x^2$ and $\sigma_y^2$ & $\sigma_{x}=\sigma_y=\tilde{\sigma}$, $\tilde{\sigma}=10$ cm\\
        $\psi(\mathbf{z})$ & polar angle of $\mathbf{z}$ & $-$\\
        $k_z$ & complex propagation constant & $-$\\
        $\theta$ & polar angle of the tag's radiation vector & $-$ \\
        $\phi$ & azimuth angle of the tag's radiation vector & $\pi/2$\\
        $F$ & number of transmit tones & $-$ \\
        $f_i$ & frequency of the $i-$th tone & $[34,54]$ GHz\\
        $\Delta f$ & inter-tone spacing & 20/$(F\!-\!1)$ (GHz) \\
        $\lambda_i$ & wavelength of the $i-$th tone & $3\times 10^8/f_i$ (m) \\
        $R_i(\theta)$ & radiation response the tag at frequency $f_i$ and polar angle $\theta$ & given in App.~\ref{appA}\\
        $s_i$ & complex radar signal transmitted at $f_i$ & $\mathbf{E}[|s_i|^2]=1$\\
        $s_i'$ & complex signal impinging the tag at $f_i$ & $-$\\
        $\tilde{s}_i''$ & noise-normalized signal received at the radar at $f_i$ & $-$ \\
        $\mathbf{h}_i(\mathbf{z})$ & channel vector between the radar and the spatial point $\mathbf{z}$ at $f_i$ & $-$\\
        $\mathbf{w}_i$ & power-normalized radar transmit precoder and receive combining vector at $f_i$ & $-$\\
        $P_i$ & transmit power of tone $f_i$ & $\propto 1/F$\\
        $\tilde{n}_i$ & normalized complex noise at the radar & $\tilde{n}_i\sim \mathcal{CN}(0,1)$\\
        $\gamma_i$ & transmit SNR of the $i-$th tone & $150$ dB $\!-\! F(\text{dB})$\\
        $K$ & number of time samples & $20$\\
        $\mathsf{L}(\varphi,\tilde{s}'')$ & likelihood function of $\varphi$ given $\tilde{s}''$ & $-$\\
        $\Phi$ & feasible search space for $\varphi$ & $-$\\
        $N$ & number of points in a brute-force search & $1000$\\
        $Q$ & number of Monte Carlo samples for the exact estimator computation & $100$\\
        $\theta_{0,i}$ & main-lobe pointing angle of $R_i(\theta)$ & \eqref{thetarad} in App.~\ref{appA}\\
        $\Theta_i,\Theta(\lambda_i)$ & half-power beamwidth of the main-lobe of $R_i(\theta)$ & \eqref{Deltarad} in App.~\ref{appA}\\
        $\kappa_i, |u_i|, v_i$ & relevant statistics corresponding to the $i-$th tone & $-$\\
        $T_s$ & sampling period & $4.63$ ns\\
        $l_1,l_2$ & length, width of the LWA & $5$ cm, $1$ cm\\
        $L$ & number of signal propagation paths & 1 \\ \bottomrule
     \end{tabular}
    \label{tab:symbols}
\end{table}
\section{System Setting}\label{system}
We consider the system illustrated in Fig.~\ref{fig1}. Specifically, an $M-$antenna radar illuminates a tag to infer its orientation angle with respect to the reference axes, assumed to be those of the radar antenna array. We focus on a two-dimensional framework to estimate the angle in the same plane as the radar for simplicity, and thus assume a uniform linear radar array (ULRA). A three-dimensional framework will be presented in future work. For convenience, but without loss of generality, we focus on the estimation of $\varphi$, hereinafter referred to as the orientation angle. 
 
Let the ULRA consecutive antenna elements be spaced $d$ meters apart, and set the position of the $m-$th element as $\mathbf{x}_m = [0, (m-1)d]^T$.
Assume that the tag's location is known, either pre-defined\footnote{For instance, tag's location may be known in advance in robotic manipulation tasks where a robotic arm aligns with objects in predefined workspaces, such as on a conveyor belt.} or as a result of prior state-of-the-art localization procedures as those in \cite{Morais.2020,Zhao.2021,Soltanaghaei.2021}. Such location knowledge is intrinsically imperfect, which is modeled using
\begin{align}    
\hat{\mathbf{z}}=\mathbf{z}+\Delta \mathbf{z},\label{zest}
\end{align}
where $\mathbf{z},\hat{\mathbf{z}}\in\mathbb{R}^2$ are the true and estimated (ULRA-perceived) tag locations, respectively. Meanwhile, $\Delta \mathbf{z}\sim \mathcal{N}(\mathbf{0},\mathbf{\Sigma})$ is the unbiased estimation error with co-variance matrix $\mathbf{\Sigma}$. For simplicity, let each element of $\mathbf{\Sigma}$ be independently  distributed, i.e., $\text{diag}(\mathbf{\Sigma})=[\sigma^2_x,\sigma^2_y]^T$. Note that the tag location angle with respect to the ULRA, denoted by $\psi$, is given by
\begin{align}
    \psi(\mathbf{z}) =\angle \mathbf{z}=\tan^{-1}(z_2/z_1). \label{psiEq}    
\end{align}
Similarly, the estimated tag location angle is $\psi(\hat{\mathbf{z}})$.
\subsection{LWA-equipped RF Sensor Tag}\label{LWAtag}
Assume a tag equipped with a LWA. Notably, the radiation pattern of an LWA can be synthesized through its complex propagation constant $k_z$, antenna geometry, and operation wavelength as exemplified in Appendix~\ref{appA}. 

Let $\theta$ and $\phi$ denote respectively a radiation vector's polar and azimuth angles from the spherical coordinate system, as illustrated in Fig.~\ref{fig1}. Meanwhile, the complex radiation gain is denoted by $R(\theta,\phi; \lambda)$. Such a frequency-dependent radiation pattern is known in advance for a given LWA implementation. We assume that such knowledge is perfect.\footnote{Imperfect radiation pattern knowledge is inevitable in practice, but its impact is left for future work. Still, note that such impairment is expected to affect the proposed orientation estimation methods MLE, P-MLE, A-MLE, and RPA, in this order, being most significant for MLE, which exploits all available information, and the least for RPA, which only exploits main-lobe radiation amplitude relative information.} Moreover, as we are focusing on a 2D analysis for simplicity, we can safely ignore the effect of the azimuth angle in our analyses. To simplify the notation, hereinafter we use $R(\theta)$ instead of $R(\theta,\phi;\lambda)$ assuming a given operation wavelength $\lambda$.

Fig.~\ref{fig:rad} illustrates the radiation pattern, both in terms of amplitude gain and phase response, for a given LWA. Herein, mm-wave operation is adopted for solution scalability since LWAs are generally several wavelengths. Note that the radiation pattern varies for different operation frequencies, which can help determine the angular direction of the tag. Moreover, a salient feature is that the main-lobe pointing angle increases with the operation frequency.
\subsection{Signal and Channel Model}\label{signalM}
Let $s$ be a narrowband signal with unit average power, i.e.,  $\mathbb{E}[|s|^2]=1$, transmitted by the ULRA towards the tag. Then, the signal impinging on the tag can be written as
\begin{align}
    s' =  \sqrt{P}\mathbf{w}^H \mathbf{h}(\mathbf{z})R\big(\theta(\varphi,\mathbf{z})\big)s,\label{eqs1}
\end{align}
where $P$ is the transmit power, $\mathbf{h}(\mathbf{z})=[h_1(\mathbf{z}), h_2(\mathbf{z}), \cdots, h_M(\mathbf{z})]^T\in\mathbb{C}^M$ is the channel vector between the ULRA and the tag (at location $\mathbf{z})$, $\mathbf{w}\in\mathbb{C}^M$ is the power-normalized transmit precoder vector, i.e., $||\mathbf{w}||^2=1$. We consider free-space line-of-sight (LOS) conditions such that the channel between $\mathbf{x}_m$ and $\mathbf{z}\in \mathbb{R}^2$ at frequency $f$ is given by
\begin{align}
    h_m(\mathbf{z})= \frac{\lambda}{4\pi||\mathbf{x}_m-\mathbf{z}||}\exp\left(-2\pi j ||\mathbf{x}_m-\mathbf{z}||/\lambda \right),\label{hieq}
\end{align}
where $\lambda=0.3/f$(GHz) is the wavelength.\footnote{Multi-path effects are numerically explored to some extent in Section~\ref{numeric}.} We focus on far-field conditions as the tag is modeled as a point, thus ignoring wavefront curvature across its surface; nevertheless, some aspects of radiative near-field propagation are captured by including per-antenna radar-to-tag paths. Meanwhile, $\theta(\varphi,\mathbf{z})$ in \eqref{eqs1} is the tag radiation polar angle experienced by the impinging signal. Note that this angle depends on the location and orientation of the tag relative to the ULRA. Fig.~\ref{fig:angles} illustrates this dependency for the six possible different cases. Assuming the sign conventions indicated in Fig.~\ref{fig:angles}, which agree with \eqref{psiEq} and Fig.~\ref{fig:rad}, we can capture the angle mathematical models of the six cases with 
\begin{align}   
\theta&(\varphi,\mathbf{z})=\mathrm{sgn}(\varphi)\frac{\pi}{2}-\psi(\mathbf{z})-\varphi,\label{thetaEq}
\end{align}
where $\mathrm{sgn}(\cdot)$ is the sign operator, i.e., $+/\!- 1$ if the input is positive/negative.
Moreover, from the left-most column cases in Fig.~\ref{fig:angles} and using geometry convention $|\theta|\le \pi/2$, we have the following system constraint 
    \begin{align}
    \left\{\begin{array}{ll}
     \psi(\mathbf{z})  + \varphi \le 0,    & \text{if}\ \psi(\mathbf{z})>0, \varphi<0,  \\
    \psi(\mathbf{z})  + \varphi \ge 0,    & \text{if}\ \psi(\mathbf{z})<0, \varphi>0.    \end{array}\right.\label{conA}
    \end{align}

\begin{figure}[t!]
    \centering    \includegraphics[width=0.91\linewidth]{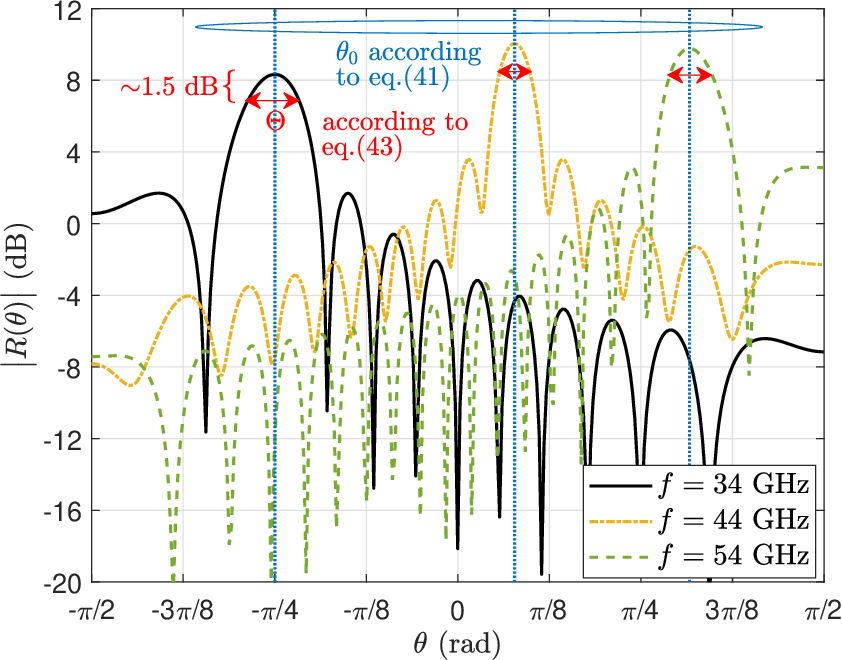}\\ \vspace{2mm}
    \ \ \includegraphics[width=0.9\linewidth]{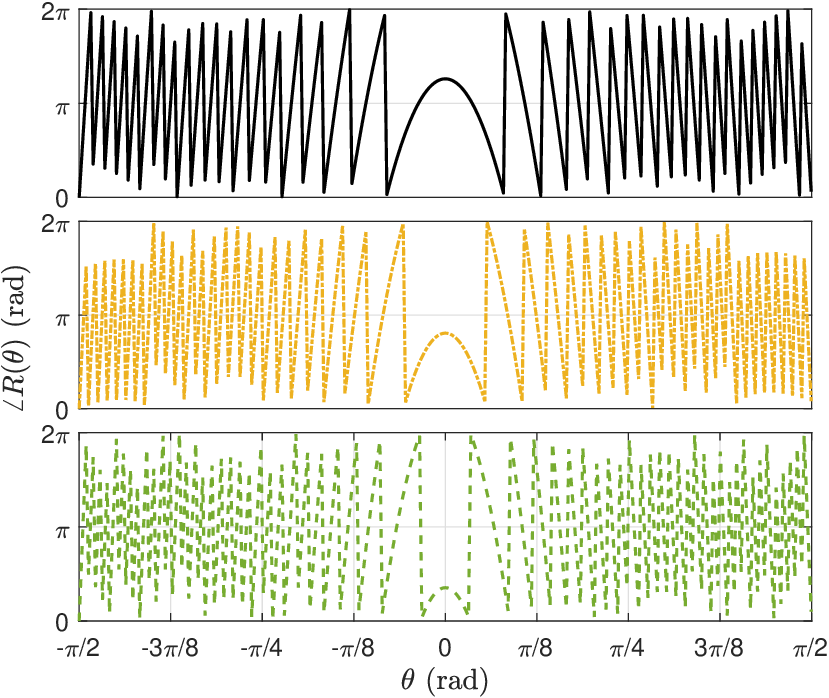}    
    \caption{a) Amplitude (top) and b) phase (bottom) radiation response of an LWA plotted according to the framework described in Appendix~\ref{appA}.}   
    \label{fig:rad}
\end{figure}
    
The tag backscatters the incident signal using a reflective load as discussed later in Section~\ref{imp}, while the reflected signal is then received at the ULRA. For this, the ULRA employs $\mathbf{w}$ as the receive combining vector. Indeed, the idea is to match both the transmit precoder and receive combiner $\mathbf{w}$ to the (imperfectly) known radar-tag channel to maximize the SNR under white noise conditions (cf. Section~\ref{sec3}). Note that this strategy is commonly adopted in radar literature when the signal's angle-of-arrival is known or estimated separately \cite{Goodman.2018,Nunez.2023}, as in our case here.

Considering perfect self-interference cancellation, as nearly achieved in practical radar systems \cite{Chang.2020,Komatsu.2021,He.2022,Abdelghaffar.2024}, the signal received at the ULRA is given by 
\begin{align}
    s'' = \mathbf{w}^H\mathbf{h}(\mathbf{z})R\big(\theta(\varphi,\mathbf{z})\big)s' + n, \label{eqs2}
\end{align}
where $n$ captures both the antenna noise at the tag and the antenna and signal processing noise at the ULRA. We model this noise as additive white Gaussian with variance $\sigma^2$, i.e., $n\sim \mathcal{CN}(0,\sigma^2)$. Note that the noise at the ULRA dominates in practice.
Substituting \eqref{zest} and \eqref{eqs1} into \eqref{eqs2}, we obtain
\begin{align}
    s'' &= \sqrt{P}(\mathbf{w}^H\mathbf{h}(\hat{\mathbf{z}}-\Delta\mathbf{z}))^2R\big(    \theta(\varphi,\hat{\mathbf{z}}-\Delta \mathbf{z})\big)^2 s + n. \label{eqs3}
\end{align}
\begin{figure}
    \centering
\includegraphics[width=0.99\linewidth]{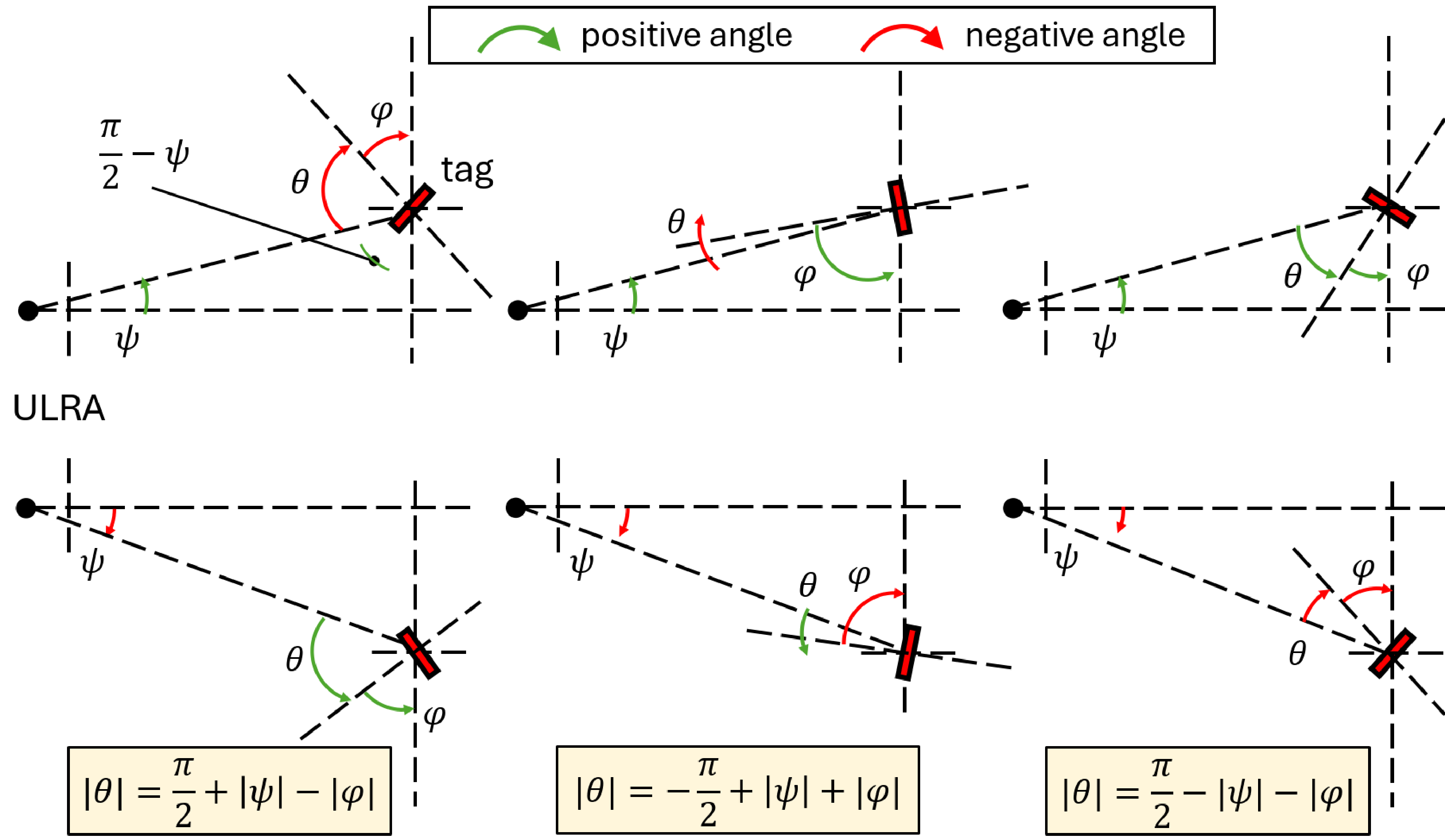}
    \caption{The six possible setup geometries and corresponding equations. Angles' signs are indicated with different colors according to the adopted convention.}
    \label{fig:angles}  
\end{figure}
\vspace{-8mm}
\subsection{Goal and Challenges}
Herein, we are concerned with finding an efficient estimator $\hat{\varphi}=g\big(s''|s, \hat{\mathbf{z}}, R(\cdot)\big)$ for $\varphi$, where $g: \mathbb{C}\rightarrow [-\pi/2, \pi/2]$. Notably, estimating $\varphi$ out of $s''$ faces some critical issues:

1) There are two randomness sources: $n$ and $\Delta \mathbf{z}$. The impact of the former, and also the latter in case the tag location is simultaneously estimated with its orientation, can be mitigated in the time domain by taking enough samples of $s''$ as each would experience different noise realizations. However, if the latter is fixed for a given measurement setup as considered here, it introduces bias unresolvable in the time domain.

2) $R(\theta)$ is not monotonic on $\theta$, but highly oscillatory, both in amplitude and phase, and thus creates ambiguities (i.e., several values of $\theta$ produce the same $R(\theta)$) that are unresolvable even when $\Delta\mathbf{z}\rightarrow \mathbf{0}$ and $n\rightarrow 0$. 

\noindent These call for further exploiting the frequency domain, e.g., using wideband transmissions as described next. Such an approach introduces measurement diversity to resolve $\Delta \mathbf{z}$. More importantly, the corresponding exploitation of the frequency scanning feature of LWAs allows radar signals to experience diverse radiation patterns, thus helping resolve ambiguities.
\section{Wideband design and estimator}\label{sec3}
Assume the radar transmits $F$ tones, which may be data-modulated as in OFDM or may not. 
Without loss of generality, let them be equally spaced such that 
\begin{align}
    f_i=f_1 + (i-1)\Delta f,\ i=1,2,\cdots, F,
\end{align}
denotes the $i-$th tone, $\Delta f$ is the inter-frequency spacing, and $F\Delta f$ is the system bandwidth. Compared to traditional radar systems, which typically occupy contiguous and wide spectral bands, such as those based on linear frequency modulation, the adopted frequency-division multiplexing is more flexible, adaptive, and easy to integrate with communication systems \cite{Ozkaptan.2018}, although at the cost of increased susceptibility to certain interference and Doppler effects, potentially high peak-to-average power ratio, and difficulties in solving range-Doppler ambiguity. In any case, radar tones can be inserted into guard bands or low-activity subcarriers of an OFDM communication frame, and/or one can leverage the precoding and receive combining structures to spatially isolate the communication and sensing signals. All in all, although the signal structure adopted herein is designed primarily for orientation estimation, it can be extended or embedded into broader joint communication and sensing system architectures with minor adaptations.

The analytical framework described previously in Section~\ref{signalM} is herein extended by including a subindex $i$ to indicate the $i-$th tone and $[k]$ to indicate the $k-$th time sample out of $K$. Specifically, we can rewrite $\eqref{eqs3}$ as
\begin{align}
    s_i''[k]\!=\!\sqrt{P_i}(\mathbf{w}_i^H\mathbf{h}_i(\hat{\mathbf{z}}\!-\!\Delta\mathbf{z}))^2R_i\big(\theta(\varphi,\hat{\mathbf{z}}\!-\!\Delta \mathbf{z})\big)^2s_i[k]\! +\! n_i[k],\label{eqs3_F}
\end{align}
for $i=1,\cdots, F$ and $k=1,\cdots, K$. The dependence of $\mathbf{h}_i$ on the frequency is clear from \eqref{hieq}.

After dividing both terms of \eqref{eqs3_F} by $\sigma$, one obtains
\begin{align}
    \tilde{s}_i''[k] = \delta_i(\varphi,\Delta\mathbf{z}) s_i[k]+\tilde{n}_i[k],\ \forall i,\ \forall k,\label{obseq}
\end{align}
where $\tilde{s}_i''[k]=s_i''[k]/\sigma$, $\tilde{n}_i[k]=n_i[k]/\sigma\sim \mathcal{CN}(0,1)$, and
\begin{align}    \delta_i(\varphi,\Delta\mathbf{z}) &\triangleq\sqrt{\gamma_i}\big(\mathbf{w}_i^H\mathbf{h}_i(\hat{\mathbf{z}}-\Delta\mathbf{z})\big)^2 R_i\big(\theta(\varphi,\hat{\mathbf{z}}\!-\!\Delta \mathbf{z})\big)^2\label{deltai}
\end{align}
with $\gamma_i\triangleq P_i/\sigma^2$ is the  transmit SNR corresponding to the $i-$th tone. Now note that for the current setup, the radar should focus the sensing beam toward the expected position of the tag, as mentioned earlier in Section~\ref{signalM}. Therefore, we adopt a maximum ratio transmit precoder, i.e., 
\begin{align}
    \mathbf{w}_i = \mathbf{h}_i(\hat{\mathbf{z}})\big/||\mathbf{h}_i(\hat{\mathbf{z}})||,\ \forall i=1,2,\cdots, F.\label{wi}
\end{align}
\subsection{MLE}
We rely on MLE to solve our problem. Since $\Delta \mathbf{z}$ is fixed and $\tilde{n}_i[k]$ realizations are independent over time and frequency-domain observations, $\tilde{s}_i''[k]$ samples are conditionally independent given 
$\Delta \mathbf{z}$ and $\tilde{n}_i[k]$. Therefore, the likelihood function $\mathsf{L}(\varphi;\tilde{s})$ given the observed data $\{\tilde{s}_i''[k]\}$ can be written as
\begin{align}
\mathsf{L}(\varphi;\tilde{s}'')&\!=\prod_{i=1}^F\prod_{k=1}^K \mathbb{E}_{\Delta \mathbf{z}}\Big[ p_{\tilde{s}''}\big(\tilde{s}_i''[k]\ \big|\ \varphi, \Delta \mathbf{z}\big)\Big]\nonumber\\ &\stackrel{(a)}{=} \prod_{i=1}^F\prod_{k=1}^K\! \mathbb{E}_{\Delta \mathbf{z}}\Big[p_{\tilde{n}}\big(\tilde{s}_i''[k]\!-\! \delta_i(\varphi,\Delta\mathbf{z})s_i[k]\ \big|\!\ \varphi, \Delta \mathbf{z}\big)\Big]\nonumber\\
 \qquad   &\stackrel{(b)}{=} \prod_{i=1}^F\prod_{k=1}^K\!\frac{1}{\pi}\mathbb{E}_{\Delta \mathbf{z}}\Big[\!\exp\Big(\!\!-\!\big(\tilde{s}_i''[k]\!-\! \delta_i(\varphi,\Delta\mathbf{z})s_i[k]\big)^H\nonumber\\
    &\qquad\qquad\qquad\ \ \ \times \big(\tilde{s}_i''[k]\!-\! \delta_i(\varphi,\Delta\mathbf{z})s_i[k]\big)\Big)\Big]\nonumber\\
    &\stackrel{(c)}{=}\! \prod_{i=1}^F\prod_{k=1}^K\!\frac{\exp(-|\tilde{s}_i''[k]|^2)}{\pi}\mathbb{E}_{\Delta \mathbf{z}}\!\Big[\!\exp\Big(\!\!-\! |s_i[k]|^2\nonumber\\
    &\ \ \times\! |\delta_i(\varphi,\!\Delta\mathbf{z})|^2\!\!+\!2\Re\Big\{\!\tilde{s}_i''[k]^*s_i[k]\delta_i(\varphi,\Delta\mathbf{z})\!\Big\}\!\Big]\!,
\end{align}
where $(a)$ comes from using $\tilde{n}_i[k]=\tilde{s}_i''[k]-\delta_i(\varphi,\Delta\mathbf{z})s_i[k]$ from \eqref{obseq}, $(b)$ from using the distribution of $\tilde{n}_i[k]$, which is a circularly symmetric complex normal random variable, and $(c)$ from expanding the squared term inside the exponential function and isolating the non-random term outside the expectation. Then, the log-likelihood function for all observations over different frequencies and time slots is given by
\begin{align}
    \ln \mathsf{L}(\varphi;\tilde{s}'')&\propto \sum_{i=1}^F\sum_{k=1}^K\ln \mathbb{E}_{\Delta \mathbf{z}}\Big[\!\exp\Big(\!\!-\! |\delta_i(\varphi,\Delta\mathbf{z})|^2|s_i[k]|^2 \nonumber\\
    &\qquad\qquad\ \ +2\Re\big\{\tilde{s}_i''[k]^*\!s_i[k]\delta_i(\varphi,\Delta\mathbf{z})\big\}\Big)\Big]  \label{logL}
\end{align}
after using $\ln \prod_i a_i=\sum_i \ln a_i$ and ignoring constant terms. 

The MLE of the surface normal angle is given by
\begin{align}
\hat{\varphi}=\arg\max_{\varphi\in\Phi} \ln \mathsf{L}(\varphi;\tilde{s}''),\label{estV}
\end{align}
where $\Phi$ is the feasible search space. Such a set comprises the whole $[-\pi/2 \pi/2]$ region excluding those intervals violating \eqref{conA}. Therefore, it can be written as
\begin{align}\label{PHI}
    \Phi = \left\{\begin{array}{ll}
     \! [-\pi/2, -\psi(\hat{\mathbf{z}})]\cup [0, \pi/2],    & \text{if}\ \psi(\hat{\mathbf{z}})>0, \\
     \! [-\pi/2, 0] \cup [-\psi(\hat{\mathbf{z}}), \pi/2],    & \text{if}\ \psi(\hat{\mathbf{z}})<0.
    \end{array}\right. 
\end{align}
Fig.~\ref{fig:logL} illustrates the log-likelihod function for ideal setups,  that is, assuming perfect tag location information and no noise impact. We can clearly observe function peaks at the ground-truth values of $\varphi$, while potential ambiguities arise in regions violating \eqref{conA}, but these can be easily discarded. Meanwhile, the other peaks in feasible regions arise due to the limited resolvability from a relatively small $F$. In fact, note that as $F$ increases, the number of high peaks at angles further from the ground-truth $\varphi$ decreases.

\begin{figure}[t!]
    \centering    \includegraphics[width=\linewidth]{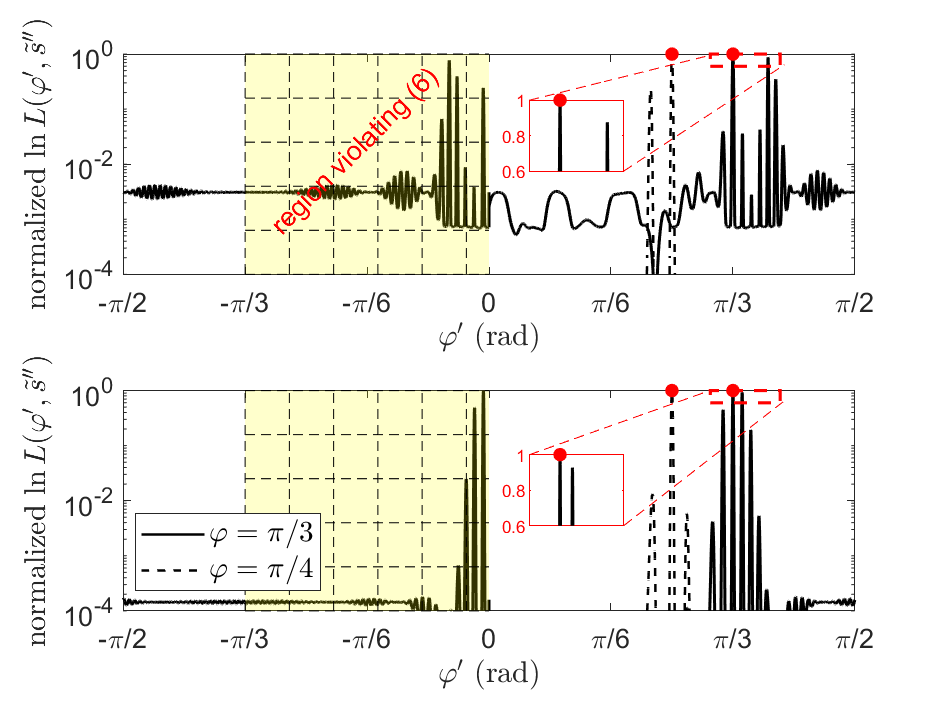}
    \caption{Normalized log-likelihood function for a ground-truth $\varphi\in\{\pi/3, \pi/4\}$, $\psi=\pi/3$, and a) $F=6$ (top) and b) $F=16$ (bottom). A fine sampling interval of $\pi/10^4$ radians (0.018$^\circ$) is used for $\varphi'$.
    We assume perfect position information, $\mathbf{\Sigma}=\mathbf{0}$, and a noiseless (ideal) scenario, $\sigma^2=0$, while the remaining simulation parameters are those assumed by default in Section~\ref{numeric}. The normalized values are obtained by applying the transformation: $(\cdot)\leftarrow ((\cdot) - \min(\cdot))/(\max(\cdot)-\min(\cdot))$.}
    \label{fig:logL}
\end{figure}

Unfortunately, there is no closed-form solution for \eqref{estV}. It is in fact intricate to compute the expectation of a highly non-linear function, which is required to evaluate \eqref{logL}.
\subsection{MLE given Perfect Tag Location Information}
Assume perfect tag location information. This is done with the hope of simplifying \eqref{estV} and obtaining a low-complexity estimator that may be used even with non-perfect location information at the cost of reduced estimation accuracy. For this, let us substitute \eqref{logL} into \eqref{estV} while letting $\Delta\mathbf{z}\rightarrow\mathbf{0}$ such that we obtain the exact-position-based MLE (P$-$MLE)
\begin{align}    \hat{\varphi}&\!=\!\arg\max_{\varphi\in\Phi}\! \sum_{i=1}^F\sum_{k=1}^K\! \Big(\!2\Re\big\{\!\tilde{s}_i''[k]^*\!s_i[k]\delta'_i(\varphi)\!\big\}\!-\! |\delta_i'(\varphi)|^2|s_i[k]|^2\!\Big)\nonumber\\
    &= \arg\max_{\varphi\in\Phi} \sum_{i=1}^F \Big(2\Re\big\{u_i\delta'_i(\varphi)\big\}-|\delta_i'(\varphi)|^2\Big),\label{estP}
\end{align}
where 
\begin{align}
    u_i\triangleq\frac{1}{v_i} \sum_{k=1}^K \tilde{s}_i''[k]^*s_i[k],\ \ \
    v_i\triangleq\sum_{k=1}^K |s_i[k]|^2, \label{ui}
\end{align}
and $\delta_i'(\varphi)\triangleq \delta_i(\varphi,\mathbf{0})$, which according to \eqref{deltai} equals
\begin{align}    \delta_i'(\varphi)=\sqrt{\gamma_i}||\mathbf{h}_i(\hat{\mathbf{z}})||^2R_i\big(\theta(\varphi,\hat{\mathbf{z}})\big)^2.\label{deltaiT}
\end{align}
Note that as per the signal model assumption at the beginning of Section~\ref{signalM} and by using a Nyquist sampling period, we have that $u_i\rightarrow \mathbb{E}[\tilde{s}_i''^*s_i]$ as $K\rightarrow\infty$.
\subsection{Complexity Analysis}\label{complexity1}
Computing \eqref{estV} and \eqref{estP} necessarily requires numerical optimization because i) the log-likelihood functions are highly non-linear and oscillatory, with many local peaks, due to the shape of $R_i(\theta)$ as illustrated in Fig.~\ref{fig:rad}; and more importantly ii) $R_i(\theta)$ is unlikely to be given in closed-form in practice, but in tabulated form. Therefore, a brute-force search over $\varphi\in[-\pi/2,\pi/2]$ seems the only reasonable approach here, and affordable as we are dealing with a single-scalar estimator. The complexity of such optimization scales with the number of points to be evaluated within the interval, denoted by $N$, and the cost of evaluating the log-likelihood function. The choice of 
$N$ depends on the desired optimization accuracy as the quantization error is given by $\pi/(N-1)$. For instance, if a resolution in the order of $1^\circ\rightarrow 0.017$ rad is required, then $N\ge \pi/0.017 + 1= 185$. Meanwhile, the cost of evaluating the log-likelihood function comprises the computation of:

\begin{itemize}
    \item the expectation of a highly non-linear function in \eqref{estV}. Monte Carlo-based integration is more appropriate here than other numeric integration approaches, as $R_i(\theta)$ is likely in tabulated form. The corresponding complexity scales with the number of samples $Q$, while the integration error scales with $1/\sqrt{Q}$.
    
    \item $M-$dimensional vector multiplications within $\delta_i$ in \eqref{estV} and $\delta_i'$ in \eqref{estP}. The corresponding cost increases linearly with $M$.
    
    \item sums over $K$, $F$, leading to $K-$, $F-$times increased complexity.
    
    \item scalar-valued functions, including $\exp(\cdot)$, $\ln(\cdot)$, $\Re\{\cdot\}$, $(\cdot)^2$, $\sqrt{\cdot}$, $(\cdot)^*$, $|\cdot|$, and scalar multiplication, which entail $\mathcal{O}(1)$.
\end{itemize}

\begin{remark}
    Based on the above, the complexity of \eqref{estV} and \eqref{estP} is 
    respectively given by $\mathcal{O}(NFKQM)$ and $\mathcal{O}(NFKM)$. Considering that $Q$ must be large, the complexity of \eqref{estP} is significantly reduced compared to \eqref{estV}. This comes, of course, with estimation performance degradation in practical setups, i.e., $\Tr(\mathbf{\Sigma})>0$, as discussed later in Section~\ref{secImp}.
\end{remark}
\section{Degenerative Impact of Tag Location Information Inaccuracy}\label{sec4}
Imperfect location information, specified by $\Delta \mathbf{z}$, affects both channel estimation and the perceived radiation pattern of the tag for a given angular configuration. These are respectively captured by $\mathbf{w}_i^H\mathbf{h}_i(\hat{\mathbf{z}}-\Delta\mathbf{z})$ and $\psi(\hat{\mathbf{z}}-\Delta\mathbf{z})$. Herein, we approximately characterize their statistics for when estimation errors are much smaller than actual distance estimations, which constitutes the case of practical interest. In the sequence, this helps to quantify the degenerative impact of imperfect location information for normal angle estimation and to propose an approximate estimator to \eqref{estV} with lower complexity.
\subsection{Statistics of $\mathbf{w}_i^H\mathbf{h}_i(\hat{\mathbf{z}}-\Delta\mathbf{z})$ and $\psi(\hat{\mathbf{z}}-\Delta\mathbf{z})$}\label{staS}
The following result reveals that channel estimation inaccuracy is approximately independent of the location estimation error statistics, i.e., $\mathbf{\Sigma}$.
\begin{lemma} \label{lemmaH}
   $\mathbf{w}_i^H \mathbf{h}_i(\hat{\mathbf{z}}\!-\!\Delta \mathbf{z})\!\sim\!      ||\mathbf{h}_i(\hat{\mathbf{z}})||\exp(j V)$, where $V$ is uni- formly random in $[0,2\pi]$, as $\Tr(\mathbf{\Sigma})/||\hat{\mathbf{z}}||^2\!\rightarrow\! 0$
     but $\Tr(\mathbf{\Sigma})\!\ne\! 0$.
\end{lemma}
\begin{proof}
    See Appendix~\ref{lemH}. \phantom\qedhere
\end{proof}
In a nutshell, the case of perfect location estimation, for which  $\Tr(\mathbf{\Sigma})=0$, differs similarly from any other scenarios with imperfect location estimation, i.e.,  $\Tr(\mathbf{\Sigma})\ne 0$. Indeed, the channel estimation inaccuracy does not increase, but remains statistically constant, with $\Tr(\mathbf{\Sigma})$.

While Lemma~\ref{lemmaH} shows that the location estimation error statistics do not impact the channel estimation accuracy, the following result reveals that they do statistically affect the accuracy of the tag's angular position estimation.
\begin{lemma} \label{lemmaPS}
As $\sigma_j/\hat{z}_i\rightarrow 0$, where $i\in\{1,2\}$ and $j\in\{x,y\}$, the distribution of $\psi(\hat{\mathbf{z}}-\Delta\mathbf{z})$ converges to
\begin{align}
\mathcal{N}\bigg(\tan^{-1}\Big(\frac{\hat{z}_2}{\hat{z}_1}\Big), 
\frac{\sigma_x^2 \hat{z}_2^2+\sigma_y^2 \hat{z}_1^2}{(\hat{z}_1^2+\hat{z}_2^2)^2}
\bigg). \label{PsiXY}
\end{align}
\end{lemma}
\begin{proof}
    See Appendix~\ref{lemPS}. \phantom\qedhere
\end{proof}
\noindent Therefore, the tag's angular position estimate is unbiased, and the estimation variance increases linearly with $\sigma_x^2,\sigma_y^2$. 
\begin{remark}
Since $|\psi(\hat{\mathbf{z}}-\Delta\mathbf{z})|\le \pi/2$, one can conclude that \eqref{PsiXY} is more accurate as $|\tan^{-1}(\hat{z}_2/\hat{z}_1)|$ is farther from $\pi/2$, i.e., when $\hat{z}_1\cancel{\lll}\hat{z}_2$ and $\hat{z}_2\cancel{\lll}\hat{z}_1$.
\end{remark}

As a corollary of Lemma~\ref{lemmaPS}, let $\sigma_x=\sigma_y=\tilde{\sigma}$, then $\psi(\hat{\mathbf{z}}-\Delta\mathbf{z})\sim \mathcal{N}(\tan^{-1}\big(\hat{z}_2/\hat{z}_1), \tilde{\sigma}^2/||\hat{\mathbf{z}}||^2\big)$. This reveals more clearly how the angular estimation accuracy decreases with the squared distance between the ULRA and the tag. 

\begin{remark}
In many envisioned setups, $\tilde{\sigma}$ is on the order of ten/hundreds of centimeters, while $||\mathbf{z}||$, thus $||\hat{\mathbf{z}}||$, may vary from a few to hundreds of meters. Therefore, $\tilde{\sigma}/||\hat{\mathbf{z}}||$ is expected to be very small, and therefore, its effect may be negligible compared to the effect of the location estimation error on the channel estimation phase as characterized in Lemma~\ref{lemmaH}.
\end{remark}
\subsection{Approximate MLE (A$-$MLE) and Complexity Analysis}\label{comp2}
Results and insights from Lemma~\ref{lemmaH} and Lemma~\ref{lemmaPS} can be exploited to achieve the following key result.
\begin{theorem}\label{the0}
    The MLE in \eqref{estV} approximates to
\begin{align}
\arg\max_{\varphi\in\Phi}\!\sum_{i=1}^F\sum_{k=1}^K \!\Big(\!I_0\big(2|\tilde{s}_i''[k]||s_i[k]||\delta_i'(\varphi)|\big)\!-\!|\delta_i'(\varphi)|^2|s_i[k]|^2\Big), \label{MLEA}
\end{align}
when tag location estimation
errors are much smaller, but with $\Tr(\mathbf{\Sigma})\ne 0$, than the actual distance estimations. Herein, $I_0(\cdot)$ is the modified Bessel function of the first kind and order 0.
\end{theorem}
\begin{proof}
See Appendix~\ref{THE0}.  \phantom\qedhere 
\end{proof}

The estimator in \eqref{MLEA} lacks the costly expectation operation in \eqref{estV}. However, it introduces the computation of $I_0(\cdot)$. Fortunately, for non-negative input values, $I_0(\cdot)$ can be easily and accurately approximated by a simpler function,\footnote{For instance, one can obtain via curve fitting $I_0(x)\approx  0.05x^{3.87}+1$ for $0\le x \le 5$, while $I_0(x)\approx 0.206\times 2.6^x$ for $x\ge 5$. This provides an average relative absolute error inferior to $3\%$.} thus avoiding incurring additional computation costs. 
\begin{remark}
    Based on the above, we can conclude that the computational cost from using \eqref{MLEA} scales as $\mathcal{O}(NFKM)$.
\end{remark}
\section{Low-Complexity Estimation}\label{sec5}
Note that computing \eqref{estP}, and even \eqref{MLEA}, although simpler than \eqref{estV} with \eqref{logL}, still requires numeric optimization due to the highly non-linear/oscillatory behavior of the corresponding objective functions. Next, we propose a lower-complexity approach that addresses this issue and discuss its key features.
\subsection{Tag Radiation Pointing Angle-based Estimator}\label{LOWsec}
Let $\theta_{0,i}$ be the main-lobe pointing angle of the radiation pattern corresponding to the $i-$th tone,\footnote{In Appendix~\ref{appA}, we use the equivalent notation $\theta_{0}(\lambda_i)$.} i.e.,
\begin{align}
    \theta_{0,i}\triangleq \arg\max_{\theta\in[-\frac{\pi}{2}, \frac{\pi}{2}]} |R_i(\theta)|.\label{theta0i}
\end{align} 
Note that $\theta_{0,1}<\theta_{0,2}<\cdots<\theta_{0,F}$, and assume
\begin{align}
\theta_{0,1}-\Theta_1/2\le \theta(\varphi,\hat{\mathbf{z}})\le \theta_{0,F}+\Theta_F/2,\label{cond}
\end{align}
where $\Theta_i$, equivalently represented as $\Theta(\lambda_i)$ in Appendix~\ref{appA}, is the main-lobe half-power beamwidth of $|R_i|$. Then, by identifying the tone corresponding to the strongest backscattered signal, one may identify the main-lobe pointing angle, and with this, $\varphi$. For this, we must first estimate $|\delta_i(\varphi)|$. 

According to the results in Lemmas~\ref{lemmaH} and \ref{lemmaPS}, the tag location information inaccuracy may be neglected here for practical setups, and thus we can focus on estimating $|\delta_i'(\varphi)|$.  Interestingly, the MLE of $\delta_i'(\varphi)$ is $u_i^*$, being $u_i$ given in \eqref{ui}. This comes from using classical results from estimation theory \cite{Kay.1993}. Meanwhile, $|u_i|$ can work as an estimator for $|\delta_i'(\varphi)|$. However, $|u_i|$ is, in general, a biased estimator of $|\delta_i'(\varphi)|$, tending to overestimate it, thus motivating the search for a more efficient estimator as given next.
\begin{theorem}\label{lem1}
    An unbiased estimator for $|\delta_i'(\varphi)|$, denoted as $\widehat{|\delta_i'(\varphi)|}$, is obtained from solving
    \begin{align}    
    \!\ _1F_1\big(\!\!-1/2,1, -|\delta_i'(\varphi)|^2v_i\big)\!=\!2|u_i|\sqrt{v_i/\pi},\label{F11}
    \end{align}
    where $_1F_1(\cdot;\cdot;\cdot)$ denotes a confluent hypergeometric function of the first kind. For the low and high SNR asymptotic regimes, the corresponding estimators can be given respectively in closed form as follows
    \begin{align}
         \widehat{|\delta_i'(\varphi)|} &= \sqrt{4|u_i|/\sqrt{\pi v_i}-2/v_i},\label{lowS}\\   \widehat{|\delta_i'(\varphi)|}&=|u_i|. \label{highS}
    \end{align}    
\end{theorem}
\begin{proof}
    See Appendix~\ref{appB}. \phantom\qedhere
\end{proof}

Note that the high-SNR asymptotic estimator \eqref{highS} matches the one coming from the absolute square transformation of the MLE of $\delta_i'(\varphi)$ as discussed before Theorem~\ref{lem1}. Moreover, such an estimator is also valid as $K\rightarrow\infty$. Therefore, numeric solvers for \eqref{F11} can exploit \eqref{highS} as a good initial guess/point. 

\begin{remark}\label{rlem1}
Very importantly, since$\ _1F_1(-1/2,1,\!-x)\!\ge\! 1$ for any $x\ge 0$, we have that solving \eqref{F11} is feasible only when $2|u_i|\sqrt{v_i/\pi}\ge 1\rightarrow |u_i| \ge \sqrt{\pi/v_i}/2$. In the case that $|u_i|<\sqrt{\pi/v_i}/2$, it is advisable to use the high-SNR expression, i.e., \eqref{highS}, as the low-SNR is also infeasible in such a range. We adopt this approach when drawing numerical results in Section~\ref{numeric}.
\end{remark}

Using \eqref{thetaEq}, \eqref{deltaiT}, and $|\varphi|\le \pi/2$, we propose
\begin{align}    
    \hat{\varphi}\!=\!\!\left\{\!\!\!\!\begin{array}{ll}
     \pi/2\!-\!\psi(\hat{\mathbf{z}})\!-\!\theta_{0,i^\star}
     &\!\!\!\!\!\text{if}\!\!\  -\!\psi(\hat{\mathbf{z}})\!\le\!\theta_{0,i^\star}\!\le\! \pi/2\!-\!\psi(\hat{\mathbf{z}}) \\
     -\!\pi/2\!\!-\!\psi(\hat{\mathbf{z}})\!\!-\!\theta_{0,i^\star}
     &\!\!\!\!\!\text{if}\!\!\  -\!\pi/2\!-\!\psi(\hat{\mathbf{z}})\!\!\le\! \theta_{0,i^\star}\!\!\!\le\!\! -\psi(\hat{\mathbf{z}})
    \end{array}\right.
    \label{varphiL}
\end{align}
as a tag radiation pointing angle-based estimator, also referred to as RPA for brevity, where 
\begin{align}
i^\star\triangleq \arg\max_{i:\ |\theta_{0,i}\!+\psi(\hat{\mathbf{z}})|\le \frac{\pi}{2}} \kappa_i\label{Istar}
\end{align}
and
\begin{align}
    \kappa_i \triangleq \frac{\widehat{|\delta_i'(\varphi)|}}{ \sqrt{\gamma_i}||\mathbf{h}_i(\hat{\mathbf{z}})||^2}.\label{ki}
\end{align}
Note that $\theta_{0,i^*}$ may diverge significantly from the ground-truth $\theta$ due to the limited frequency diversity, i.e., relatively small $F$, causing a kind of quantization error, and/or due to a limited angular resolution, i.e., relatively large (small) $\theta_{0,1}$ ($\theta_{0,F}$) when $i^*=1\ (F)$. We explore this further in the following. 

\begin{theorem}\label{the:t3}
    The feasible estimation range for $\varphi$ when using RPA is approximately given by 
    \begin{align}
        &\big[\mathrm{sgn}(\varphi)\pi/2-\psi(\hat{\mathbf{z}})-\theta_{0,F}-\Theta_F/2,\nonumber\\
        &\qquad\qquad\mathrm{sgn}(\varphi)\pi/2-\psi(\hat{\mathbf{z}})-\theta_{0,1}+\Theta_1/2\big].\label{the3}
    \end{align}
\end{theorem}
\begin{table*}[t!]
    \centering
    \caption{Key features of the proposed estimators}
    \begin{tabular}{p{1.4cm}|p{3.1cm}|p{3.6cm}|p{3.5cm}|p{4.6cm}} \toprule
         & \multicolumn{3}{|c|}{\textbf{MLE-based}} &   \\
       \midrule
       \textbf{estimator} & exact MLE, \eqref{estV}  & exact-position-based MLE, \eqref{estP}  & approximate MLE, \eqref{MLEA} & (based on) radiation pointing angle, \eqref{varphiL} \\
       \textbf{acronym} & MLE & P$-$MLE & A$-$MLE & RPA \\
       \textbf{assumptions} & $-$ & $\Delta\mathbf{z}=0$ & $\Delta\mathbf{z}/||\hat{\mathbf{z}}||\approx 0$ (or $\tilde{\sigma}/||\hat{\mathbf{z}}||\approx 0$) & $\Delta\mathbf{z}/||\hat{\mathbf{z}}||\approx 0$ (or $\tilde{\sigma}/||\hat{\mathbf{z}}||\approx 0$) \\
       \textbf{input data} & $\tilde{s}'', \hat{\mathbf{z}}, \{\gamma_i\}, \mathbf{\Sigma}, \{R_i(\theta)\}_{\forall\theta}$  & $\tilde{s}'', \hat{\mathbf{z}}, \{\gamma_i\}, \{R_i(\theta)\}_{\forall\theta}$ & $\tilde{s}'', \hat{\mathbf{z}}, \{\gamma_i\}, \{|R_i(\theta)|\}_{\forall\theta}$ & $\tilde{s}'', \hat{\mathbf{z}}, \{\gamma_i\}, \{\theta_{0,i}\}$\\
       \textbf{workflow} & substitute \eqref{hieq}, \eqref{thetaEq}, \eqref{deltai}, \eqref{wi}  into  \eqref{logL} together with input data, and compute \eqref{estV} using   \eqref{logL} and \eqref{PHI} 
       &  substitute  \eqref{hieq}, \eqref{thetaEq}, \eqref{PHI}, \eqref{ui}, \eqref{deltaiT}  into  \eqref{estP} and compute it using the input data
       &  substitute  \eqref{hieq}, \eqref{thetaEq}, \eqref{PHI}, \eqref{deltaiT}  into \eqref{MLEA} and compute it using the input data
       & obtain $\widehat{|\delta_i'(\varphi)|}$ from Theorem~\ref{lem1} (and using \eqref{ui} and input data), substitute it into \eqref{ki} together with  \eqref{hieq}, then use \eqref{psiEq}, \eqref{Istar}, and \eqref{ki} to compute 
       \eqref{varphiL} \\
       \textbf{complexity} & $\mathcal{O}(NFKQM)$ & $\mathcal{O}(NFKM)$ & $\mathcal{O}(NFKM)$ & $\mathcal{O}(F(K+M))$ \\ \bottomrule
    \end{tabular}
    \label{tab:complexity}
\end{table*}
\begin{proof}
Let us depart from \eqref{varphiL}. Assume first that $-\psi(\hat{\mathbf{z}})\le \theta_{0,i^\star}\le \pi/2-\psi(\hat{\mathbf{z}})$, then $\hat{\varphi}=\pi/2-\psi(\hat{\mathbf{z}})-\theta_{0,i^\star}\ge 0$, and using the bounds in \eqref{cond}, one gets $\pi/2-\psi(\hat{\mathbf{z}})-\theta_{0,F}-\Theta_F/2\le \hat{\varphi}\le \pi/2-\psi(\hat{\mathbf{z}})-\theta_{0,1}+\Theta_1/2$. Meanwhile, assuming $-\pi/2-\psi(\hat{\mathbf{z}})\le \theta_{0,i^\star}\le -\psi(\hat{\mathbf{z}})$ in \eqref{varphiL}, one gets $\hat{\varphi}=-\pi/2-\psi(\hat{\mathbf{z}})-\theta_{0,i^\star}\le 0$. Using again the bounds in \eqref{cond}, one obtains $-\pi/2-\psi(\hat{\mathbf{z}})-\theta_{0,F}-\Theta_F/2\le \hat{\varphi}\le -\pi/2-\psi(\hat{\mathbf{z}})-\theta_{0,1}+\Theta_1/2$ for this case. Combining these results, one achieves \eqref{the3}. 
\end{proof}

\begin{remark}\label{RW}
By substracting both bounds in \eqref{the3}, one obtains that the feasible angular estimation range is approximately given by $\theta_{0,F}-\theta_{0,1} + (\Theta_{F}-\Theta_{1})/2$. Again, as the set of potential estimate values is discrete under RPA, we can expect a kind of quantization error that decreases proportionally to $F$.
\end{remark}
\subsection{Estimation Accuracy Performance Trends}
The estimation dispersion of $\kappa_i$, as a measure of estimation inaccuracy, increases with that of $|u_i|$. The specific relationship between $\kappa_i$ is quite involved as captured by \eqref{F11}, but much simpler in the asymptotic regimes as captured by \eqref{lowS} and \eqref{highS}. Specifically, $\kappa_i$ increases with the square root and linearly with $|u_i|$ in the low and high-SNR regimes, respectively. 
In any case, the variance of $|u_i|$, which is characterized next, is useful as a measure of estimation dispersion of $\kappa_i$, matching the exact dispersion in the case of high-SNR.
\begin{theorem}\label{the1}
     The variance of $|u_i|$ is given by
    \begin{align}
    \mathbb{V}[|u_i|]\!=\!|\delta_i'(\varphi)|^2\!\!+\!\frac{1}{v_i}\!\Big(1\!-\!\frac{\pi}{4}\!\!\ _1F_1\Big(\!\!-\!\frac{1}{2},1,\!-|\delta_i'(\varphi)|^2v_i\Big)\Big), \label{Uvar}
\end{align}
which converges as $K$ increases to
\begin{align}
    \mathbb{V}[|u_i|]\!=\!|\delta_i'(\varphi)|^2\!+\!\frac{1\!-\!\frac{1}{2}\ |\delta_i'(\varphi)|\sqrt{K\pi}}{K}.\label{UvarB}
\end{align}
\end{theorem}
\begin{proof}
From the distribution of $|u_i|$ given in Appendix~\ref{appB}, we obtain \eqref{Uvar}. Since $v_i\rightarrow K\mathbb{E}[ |s_i[k]|^2]\rightarrow K$ as $K\rightarrow\infty$ , while using $_1F_1(-1/2,1,x)\rightarrow 2\sqrt{-\frac{x}{\pi}}\ \ \text{as}\ \ x\rightarrow -\infty$,
we have
\begin{align}
    \ _1F_1\Big(\!\!-\!\frac{1}{2},1,-|\delta_i'(\varphi)|^2 K\!\Big)\!\!\rightarrow\! 2|\delta_i'(\varphi)|\sqrt{\frac{K}{\pi}}\ \text{as}\ K\rightarrow \infty. \label{asF}
\end{align}
Substituting \eqref{asF} into \eqref{Uvar}, we obtain \eqref{UvarB}.
\end{proof}
\begin{remark}\label{re7}
    Note from \eqref{UvarB} that $\lim_{K\rightarrow\infty}\mathbb{V}[|u_i|]=|\delta_i'(\varphi)|$, while $\lim_{F\rightarrow\infty}\mathbb{V}[|u_i|]=1/K$ since $|\delta_i'(\varphi)|\propto \sqrt{\gamma_i}$ while $\gamma_i$ decreases with $F$ given a fixed total transmit power budget. Thus, arbitrarily increasing estimation accuracy requires increasing both $F$ and $K$. Also, in high-SNR setups, using at least a large $F$ seems compulsory.
\end{remark}
\subsection{Complexity Analysis}\label{comp3}
The computational complexity of the estimator in \eqref{varphiL} is $\mathcal{O}(F(K+C+M))$. This is because, for each tone, one must compute $\sum_{k=1}^K |s_i[k]|$ and $|u_i|$, whose complexity scales with $K$, and then solve \eqref{F11}, assumed with complexity cost $C$, and evaluate \eqref{ki}, which requires computing the norm of an $M-$dimensional vector. Obtaining $i^*$ and evaluating \eqref{varphiL} incurs a complexity respectively of $\mathcal{O}(F)$ and $\mathcal{O}(1)$, thus negligible compared to the previous operations. In fact, $\theta_{0,i}$ specified in \eqref{theta0i} is known in advance for any given tag. 

Due to its quadratic convergence, the Newton-Raphson method is a natural conventional technique for efficiently solving \eqref{F11}. This method requires evaluating the confluent hypergeometric function and its derivative, each typically with complexity scaling with the number of terms $N'$, in their series expansion needed for the required precision.
Meanwhile, given the quadratic convergence, the number of iterations required to achieve $D$ digits of precision scales proportionally to $\ln D$. Thus, the overall computational complexity scales as $C=N'\ln D$.

Since $C$ from solving \eqref{F11} can grow large for an arbitrary accuracy requirement, we alternatively propose 
\begin{align}
    \widehat{|\delta_i'(\varphi)|}=    \sqrt{\Big(0.1944\big(|u_i|\sqrt{v_i/\pi}+0.067\big)^2-1\Big)\Big/v_i}
\end{align}
as a low-complexity closed-form approximate solution. This comes from using  
\begin{align}
    _1F_1(-1/2,1,-x)\approx 1.134\sqrt{x+1}\!-\!0.134,\ \text{for}\ x\ge 0,\label{F11A}
\end{align}
obtained via curve fitting. Via numerical analysis, we found that the achieved fitting accuracy is noticeably good, with a relative approximation error always below $1.7\%$.

\begin{remark}
The complexity of evaluating \eqref{F11A} only scales with $K$. Therefore, with its use, the estimator in \eqref{varphiL} achieves a computational complexity of $\mathcal{O}(F(K+M))$, significantly lower than that of all the previous estimators.
\end{remark}
\section{Implementation Considerations}\label{imp}
As shown at the right-bottom of Fig.~\ref{fig1}, the backscattering tag can be realized by connecting the LWA to a reflective load, which might require an interface network to match impedances and/or properly route signals, and ensure efficient reflection and reradiation. These elements are mostly low-complexity/cost, thus suitable for tag applications. As per the signal processing techniques derived in Sections~\ref{sec3}-\ref{sec5} for estimating the orientation of LWA-equipped tags, they are summarized in Table~\ref{tab:complexity} along with their key features, especially exploited data, estimation workflow, and complexity. Independently of the specific estimator, there are two primary methods for transmitting the tones as discussed below.
\subsubsection{Simultaneous Transmissions}
This constitutes the approach considered in earlier sections. This requires the system to handle wideband transmissions, as multiple and widely-spaced frequencies are used concurrently. The main appeals of this approach are the short sensing time, which scales linearly with $K$, and OFDM compatibility, which fits joint communication and sensing frameworks. Meanwhile, the challenges are related to advanced hardware requirements, including a wideband full-duplex transceiver capable of generating, sending, and receiving without significant effective self-interference. This demands sophisticated filtering and separation techniques for mitigating intermodulation and cross-talk interference among tones.\footnote{Readers can refer to \cite{Abdelghaffar.2024} for state-of-the-art discussions on full-duplex transceivers and operation, including challenges and enabling techniques.} The number of tones $F$ is limited as $\Delta f$ cannot be smaller than the inverse of the sampling period, although this may not be a stringent limitation in systems with sufficiently large $f_F-f_1$.
\subsubsection{Frequency Sweeping Transmissions}
The tones are transmitted in a time-division manner. The proposed signal processing techniques still apply, but they need to account for the sequential nature of the data collection. Indeed, the sensing time now scales linearly with $KF$, thus larger than for the simultaneous tone transmissions approach, but the number of tones $F$ is not limited. Although the system can now be narrow-band, it must be capable of switching frequencies rapidly and accurately, while ensuring precise synchronization between transmission and reception times for each tone. Still, the system design is simpler as it deals with one frequency at a time, reducing the complexity of signal processing and hardware requirements, making it more accessible for lower-cost implementations.
\section{Numerical Performance Analysis}\label{numeric}
In this section, we numerically assess the performance of the derived estimators under practical system conditions using custom-developed MATLAB scripts, which implement the discussed analytical models and estimation procedures. We show averaged results in terms of absolute estimation error, i.e., $\mathbb{E}_{\hat{\mathbf{z}}}[|\varphi-\hat{\varphi}(\hat{\mathbf{z}})|]$, computed over 100 Monte Carlo system realizations. Results are given in degrees to ease interpretation.

\begin{figure}
    \centering    \includegraphics[width=0.9\linewidth]{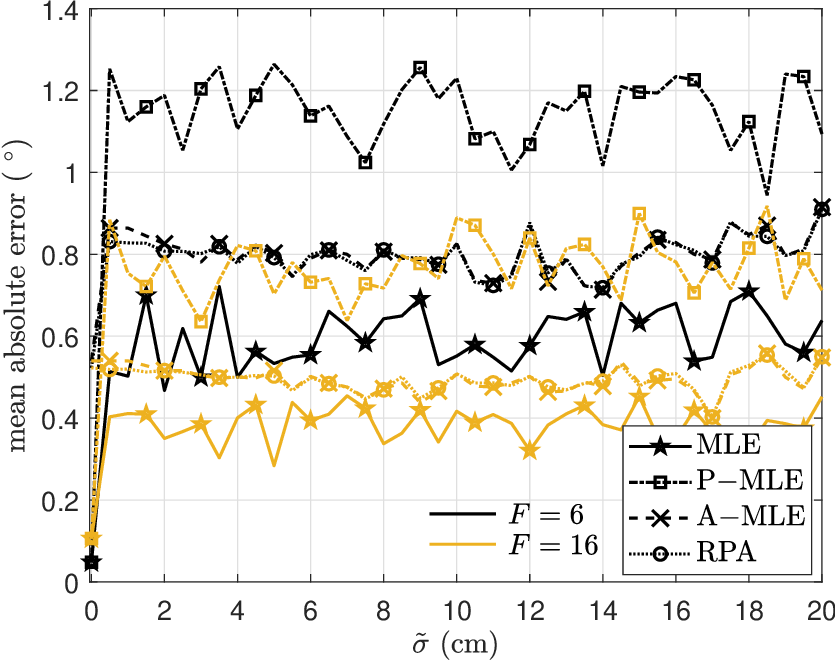}    
    \caption{Mean absolute estimation error of $\varphi$ as a function of the standard deviation of the tag location estimation error for $F\in\{6,16\}$, $\psi=60^\circ$, and $\varphi=45^\circ$ as the ground-truth.}
    \label{fig:deltaZ}
\end{figure}
The radiation profile of the LWA is computed according to the framework in Appendix~\ref{appA}.
Unless stated otherwise, we assume the system parameter values as indicated in the right-most column of Table~\ref{tab:symbols}. Note that we set $f_1=34$ GHz and $f_F=54$ GHz with $F\ge 2$, and $\Delta f=(f_F-f_1)/(F-1)$, thus targeting a mm-wave implementation.\footnote{A large bandwidth is adopted here to understand the limits of estimation accuracy. In fact, radars with up to 40 GHz bandwidth have been reported in the literature, e.g. \cite{Welp.2020}. Note that reducing the bandwidth leads to a narrower angular estimation range but better estimation accuracy of the feasible angles given a fixed $F$.} Still, a small number of antennas $M=4$ is used for simplicity, while this is nevertheless compatible with many mm-wave radar platforms, e.g., TI's IWR/AWR series. Moreover, we assume unmodulated tone transmissions, such that $s_i[k]=\exp(2\pi j f_i T_s (k-1))$,  where $T_s$ is the sampling period. We set $T_s=1/(4f_F)\approx 4.63$~ns. Also, let $P_i=P_T/F, \forall i$, where $P_T$ is the total power budget, thus $\gamma_i=(P_T/\sigma^2)/F$, and we set $P_T/\sigma^2=150$~dB.\footnote{Assuming a conservative bandwidth for signal processing (after downconversion) of $\sim\!200$ kHz, one would require $P_T\approx 1$ W to achieve this.} We set $\tilde{\sigma}=10$ cm by default as this is a well-known target for next-generation cellular-based positioning systems \cite{Morais.2020,Lopez.2023}. Finally, note that far-field propagation conditions hold in all the simulated scenarios throughout the section.
\begin{figure}[t!]
    \centering    \includegraphics[width=0.9\linewidth]{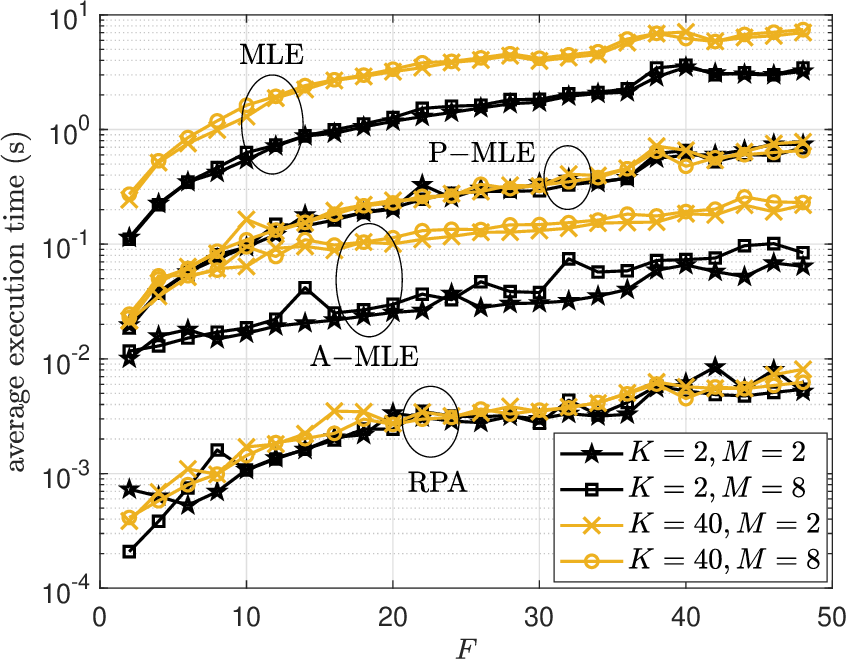}    
    \caption{Average estimation runtime as a function of $F$ for $K\in\{2,40\}$, $M\in\{2,8\}$, $\psi=60^\circ$, and $\varphi=45^\circ$ as the ground-truth.}
    \label{fig:complx}
\end{figure}
\subsection{Impact of Tag Location Estimation Error}
\label{secImp}
\begin{figure*}[t!]
    \centering 
    MLE \qquad\qquad\qquad\qquad\qquad\qquad\ \ \ \ A$-$MLE \qquad\qquad\qquad\qquad\qquad\qquad\ \ \ RPA\ \ \ \ \\
    \vspace{-16mm}
    \begin{minipage}[t]{0.95\textwidth}
    \begin{minipage}{0.36\textwidth}{\includegraphics[width=\textwidth]{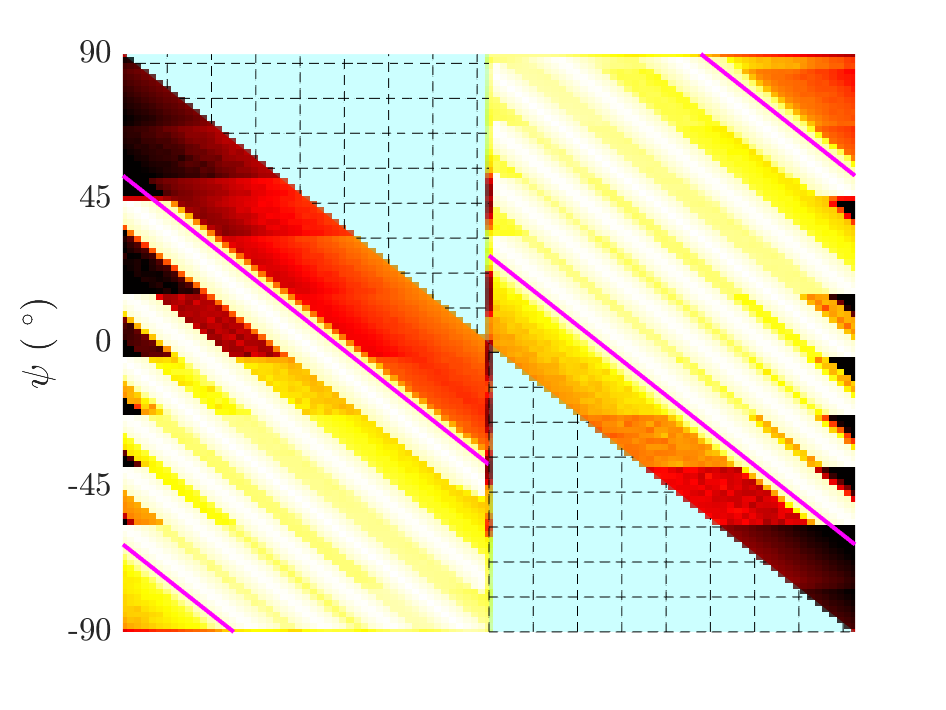}}
    \end{minipage}\!\!\!\!\!\!\!\!\!
    \begin{minipage}{0.36\textwidth}  {\includegraphics[width=\textwidth]{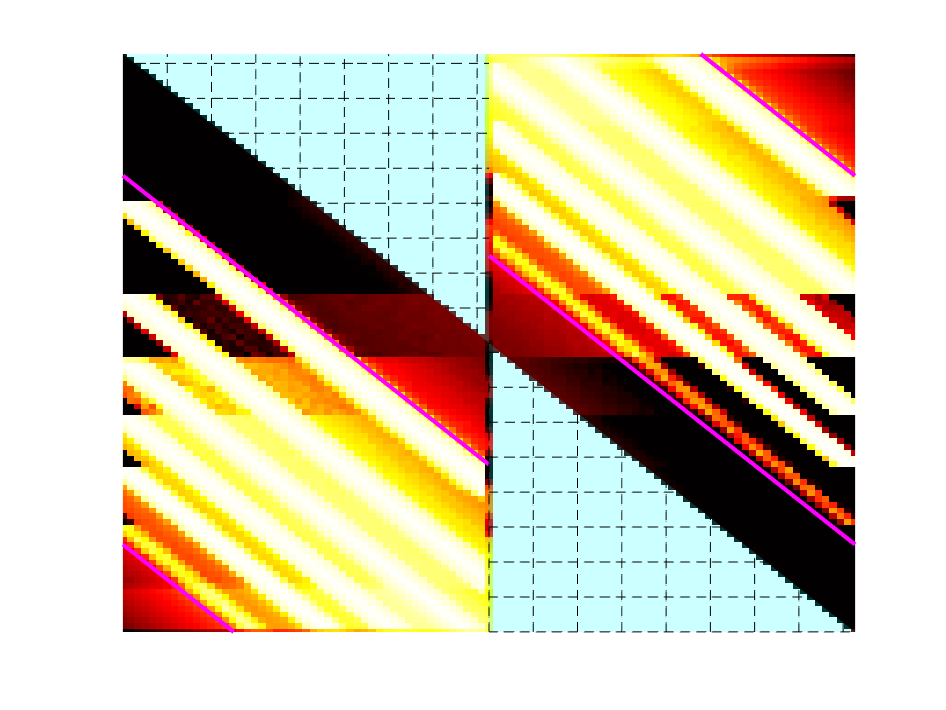}}
    \end{minipage}\!\!\!\!\!\!\!\!\!
    \begin{minipage}{0.36\textwidth}{\includegraphics[width=\textwidth]{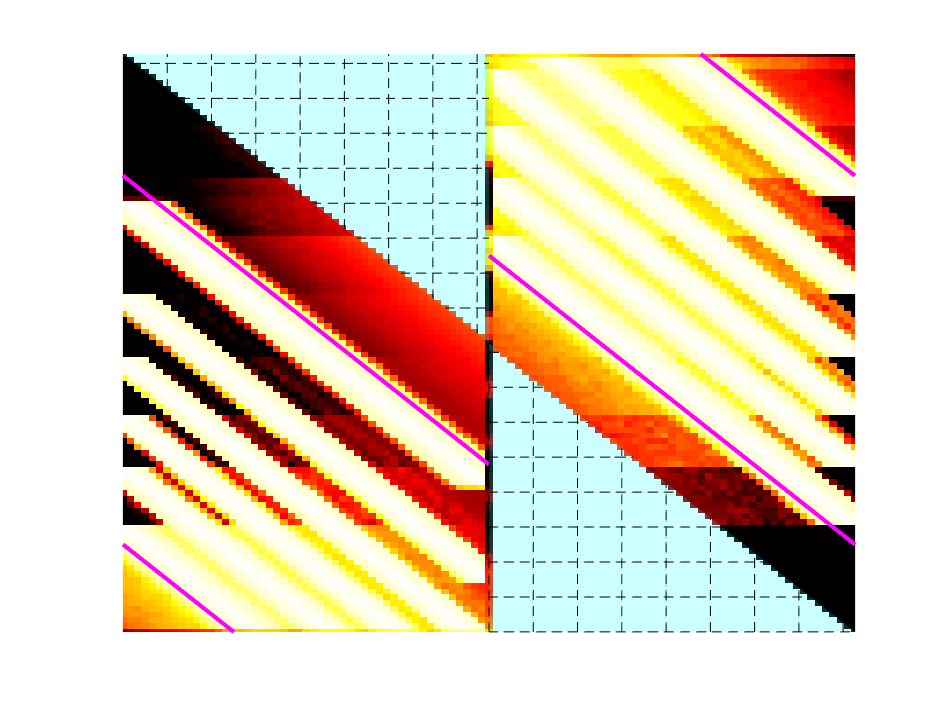}}
    \end{minipage}    
    \vskip\baselineskip 
    \vspace{-8mm}
    \begin{minipage}{0.36\textwidth}{\includegraphics[width=\textwidth]{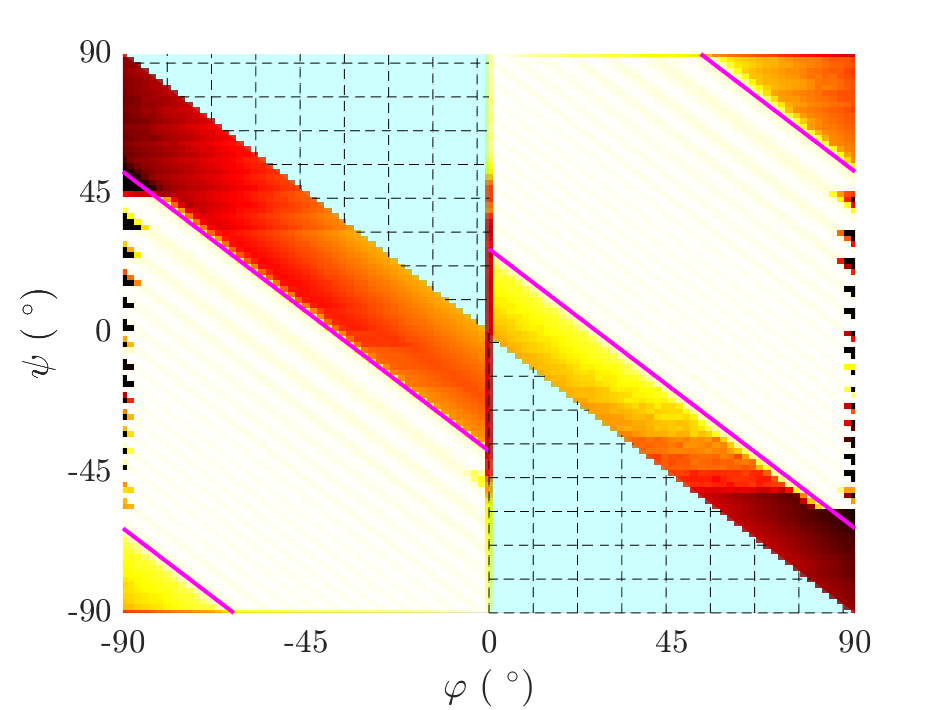}}
    \end{minipage}\!\!\!\!\!\!\!\!\!
    \begin{minipage}{0.36\textwidth}{\includegraphics[width=\textwidth]{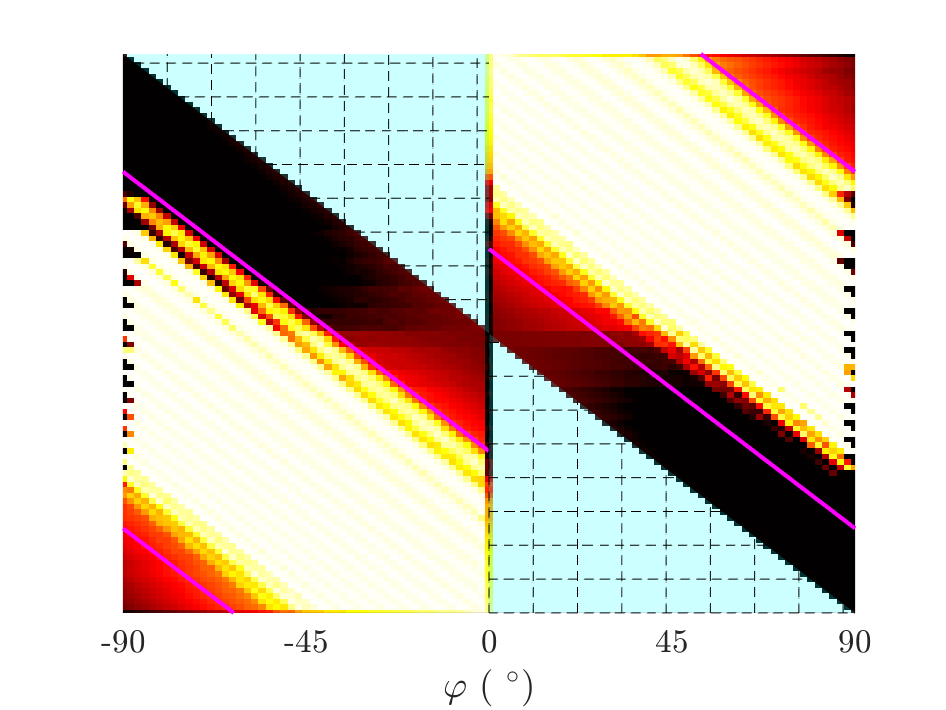}}
    \end{minipage}\!\!\!\!\!\!\!\!\!
    \begin{minipage}{0.36\textwidth}{\includegraphics[width=\textwidth]{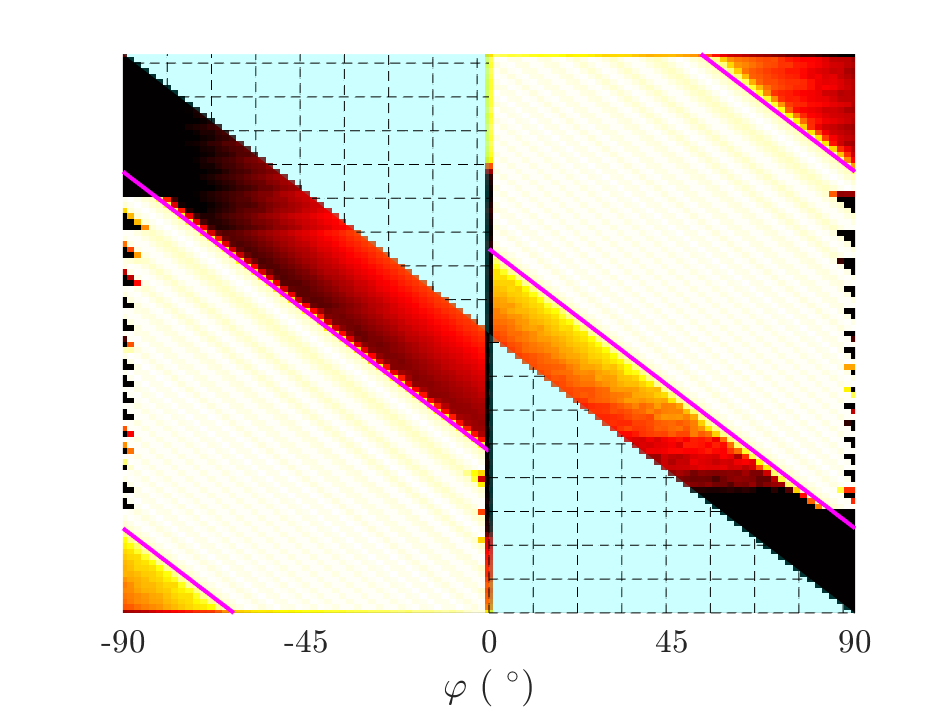}}
    \end{minipage}
    \end{minipage}
    \!\!
    \begin{minipage}{0.04\textwidth} \ \\ \ \\ \ \\ \ \\ \ \\ \ \\ \ \\ \ \\ \ \\    \includegraphics[width=1.05\textwidth]{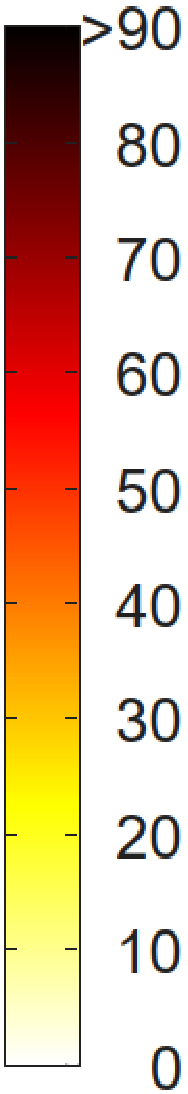}
    \end{minipage} 
    \caption{Heatmap of the average mean estimation error of $\varphi$ as a function of the tag angular position and target normal angle. We adopt
    $F=6$ and $F=16$ in the figures in the first and second rows, respectively. We illustrate with magenta lines the limits of the regions specified in \eqref{the3}, and color in a cyan pattern the (ambiguity) region violating \eqref{conA}.}    \label{fig:HeatMap1}
\end{figure*}
Fig.~\ref{fig:deltaZ} captures the mean absolute orientation estimation error as a function of $\tilde{\sigma}$ for $F\in\{6,16\}$, $\psi=60^\circ$, and $\varphi=45^\circ$ as the ground-truth. The results here corroborate our findings in Section~\ref{staS}: the estimation performance is approximately independent of the location estimation error statistics, i.e., $\tilde{\sigma}$, for practical settings. The latter implies setups with imperfect tag location information and relatively small $\tilde{\sigma}/||\mathbf{z}||$. Indeed, note that MLE and P$-$MLE perform extremely well for $\tilde{\sigma}=0$, but similarly worse for any other $\tilde{\sigma}>0$. Meanwhile, since A$-$MLE exploits Lemma~\ref{lemmaH} assuming $\tilde{\sigma}>0$, and RPA is agnostic of location estimation error statistics, they do not offer as good performance for $\tilde{\sigma}=0$. However, $\tilde{\sigma}=0$ is not achievable in practice, tilting the scale in favor of A$-$MLE and RPA, which are more affordable. In fact,  all the estimators perform similarly for $\tilde{\sigma}>0$, although P$-$MLE does it relatively worse as it assumes no location estimation error and thus it does not exploit key related statistics, i.e., $V$ in Lemma~\ref{lemmaH}. We can observe that A$-$MLE and RPA perform tightly close and not far from MLE, evincing their robust design framework, especially that of RPA, which achieves this with extremely low complexity.
\subsection{On the Estimation Complexity}\label{cmplx}
Herein, we assess the runtime of the four proposed estimators to complement the theoretical complexity discussions throughout Sections~\ref{complexity1}, \ref{comp2}, and \ref{comp3}. Fig.~\ref{fig:complx} shows the average estimation runtime as a function of $F$ for different combinations of $M$ and $K$. Although the y-axis is presented in logarithmic scale to facilitate visualization, we want to emphasize that the curves in linear scale are approximately linear on $F$ in all the cases, as predicted by our complexity analysis. Meanwhile, the runtime impact of $M$ and $K$ is shown to be less pronounced because many matrix operations involving these parameters are highly parallelizable and efficiently handled by MATLAB's vectorized execution engine, though they still incur greater computational load. Only for MLE and A$-$MLE, a longer sensing time $K$ somewhat still affects the overall curves' slope. These trends empirically confirm the relative complexity ranking discussed in the paper: RPA is consistently the most efficient estimator, followed by A$-$MLE, then P$-$MLE, with MLE being the most computationally intensive. Indeed, for the selected configuration setup, A$-$MLE, P$-$MLE, and MLE require $10-80\times$, $60-120\times$, and $400-1200\times$ more time than RPA, respectively.

To facilitate discussions, we remove P$-$MLE from the next numerical performance assessments. Recall that P$-$MLE underperforms all other estimators in any practical setting and with a complexity never lower than A$-$MLE and RPA.
\subsection{On the Estimation Angular Range and Accuracy}\label{rBW}
Fig.~\ref{fig:HeatMap1} depicts the average mean estimation error of $\varphi$ for all possible ground-truth pairs  $(\varphi,\psi)$ and $F\in\{6,16\}$. We color in cyan the (ambiguity) region violating \eqref{conA}, which can be discarded. Also, we illustrate with magenta lines the limits of the regions specified in \eqref{the3}, which notably apply accurately not only to RPA but also to MLE and less tightly to A$-$MLE. The latter suggests that the tag localization estimation error affects A$-$MLE more than RPA (and MLE, of course) in the boundary of \eqref{the3}. Indeed, it makes sense that the feasible estimation angular range for every estimator considered in this paper somewhat agrees with \eqref{the3}. This is because in such a region, every $\varphi$ corresponds to a $\theta$ between $\theta_{0,1}-\Theta_1/2$ and $\theta_{0,F}+\Theta_F/2$, wherein all the pointing angles of the main-lobes of the LWA radiation pattern, $\{\theta_{0,i}\}$, reside. As shown in Fig.~\ref{fig:rad}a, there is no radiation amplitude ambiguity close to these pointing angles, thus accurately retrieving the ground-true $\theta$ (and then $\varphi$) is only possible here. In fact, the diagonal bright lines in Fig.~\ref{fig:HeatMap1} correspond to pairs $(\varphi,\psi)$ that lead to $\theta$ values close to $\theta_{0,i}$ for some $i=1,2,\cdots,F$. Note that the phase response of the LWA radiation pattern is expected with little, if any, impact on a proper estimator design due to no phase ambiguity-free region as illustrated in Fig.~\ref{fig:rad}b. This is somewhat verified here also by noticing the similar performance attained by MLE, RPA, and A$-$MLE to some extent, being A$-$MLE and RPA only exploiting $\{|R_i(\theta)|\}$-related information.

Regardless of the feasible angular estimation region's coverage, some angular configurations always cause poor estimator performance. That is the case of the ground-truth pairs $(\psi,\varphi)$ close to $(90^\circ, \pm 90^\circ)$, $(-90^\circ, \pm 90^\circ)$, $(\pm 45^\circ, \mp 45^\circ)$, and $(0^\circ,0^\circ)$ since the signals impinging the tag arrive with $\theta\rightarrow\pm 90^\circ$. In addition to these, another critical configuration is that where $\varphi \approx -\psi(\hat{\mathbf{z}})$ since a feasible/unfeasible angular configuration according to \eqref{conA} may be identified as unfeasible/feasible due to tag localization estimation errors, thus affecting the whole orientation estimation process. Back to Fig.~\ref{fig:HeatMap1}, observe that the feasible angular estimation regions light up as $F$ increases from 6 to 16, as the number of bright lines matches $2F$ ($F$ for each $\varphi>0$ and $\varphi<0$, each corresponding to a certain $\theta_{0,i}$), although not exactly in the case of A$-$MLE, as its accuracy is more affected in the boundary of the feasibility regions \eqref{the3} as discussed earlier. In any case, an increase in $F$ leads to better angular coverage in terms of estimation error. In fact, when computing the mean absolute estimation error of the lower 75-th percentile of the feasible angular estimation range, i.e., lower-Q3 mean, that is the mean estimation error for $\varphi$ throughout the 75\% best pairs $(\psi,\varphi)$ in the so-called feasible angular estimation regions,\footnote{This is done to discard the effect of the outliers, including several angular configurations encompassing values of $\varphi$ close to $\pm 90^\circ$.} we obtain values close to $3.6^\circ, 8.8^\circ$, and $5.0^\circ$ respectively for MLE, A$-$MLE, and RPA with $F=6$, while these reduce to $0.9^\circ, 2.1^\circ$, and $1.4^\circ$ for $F=16$. Note that a more homogeneous coverage can be achieved, at least for MLE and RPA, by selecting the tones such that $\{\theta_{0,i}\}$ are equally spaced instead of using equally-spaced tones as done here.

Next, we delve into angular coverage analysis while focusing only on the performance of RPA given its near-optimality and low complexity. This facilitates obtaining results even for a large $F$ and having more focused illustrations and discussions.
\subsection{Angular Coverage Accuracy}\label{AngCov}
\begin{figure}[t!]
    \centering    \includegraphics[width=0.9\linewidth]{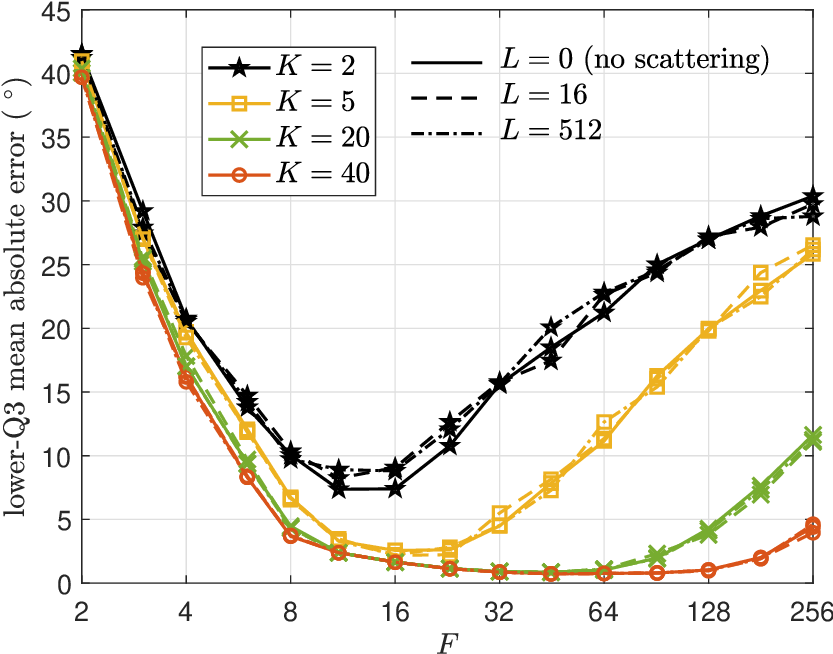}
    \caption{Lower-Q3 mean absolute estimation error of $\varphi$~as a function $F$ for RPA with $K\in\{2,5,20,40\}$ and $L\in\{0,16,512\}$. We compute the average performance over every ground-truth $\varphi$ within the feasible estimation range \eqref{the3}.}
    \label{fig:FigFK}
\end{figure}
Fig.~\ref{fig:FigFK} shows the lower-Q3 mean absolute estimation error within the feasible angular estimation range (according to \eqref{the3}) as a function of $F$, achieved by RPA for $K\in\{2,5,20,40\}$. It also depicts the impact of multi-path propagation by considering that the sensing signal from the direct LOS link is affected by $L$ additional propagation paths. For this, we deploy $L$ scatterers randomly located inside a Cassini oval with foci at the radar and tag such that, compared to the LOS link, the power propagation gain and the phase shift of each scattered signal are at least $10^{-5}\%$ and $\lambda_1/4$, respectively. This ensures that $L$ relevant scattering paths are always present. Note that the multi-path effect does not critically affect the estimation accuracy, being almost unnoticeable even for the extreme $(K=2, L=512)$ configuration. This is somewhat expected due to the high propagation losses at mm-wave frequencies. Moreover, increasing $K$ improves estimation accuracy and makes the system more robust against multi-path due to the known averaging-out effect. One of the main highlights from Fig.~\ref{fig:FigFK} is that given a fixed power and time budget, increasing $F$ unbounded is not convenient as the transmit SNR of each tone is increasingly affected and this can eventually deteriorate the estimation accuracy. This complements our insights in Remark~\ref{re7}. In fact, there is an optimum $F$, which increases with $K$. For example, for the setup simulated to draw Fig.~\ref{fig:FigFK}, the optimum $F$ is $12$, $16$, $45$, and $52$ for $K=2$, $K=5$, $K=20$, and $K=40$, respectively, leading to an orientation estimation error about $7.8^\circ, 2.2^\circ$, $0.9^\circ$, and $0.7^\circ$.  This can be explained by noting that each tone transmission lasts longer by increasing $K$, which compensates for the reduced transmit SNR as $F$ increases, i.e., effectively reducing sensing signal power but maintaining its energy. 
\begin{figure}[t!]
    \centering    \includegraphics[width=0.9\linewidth]{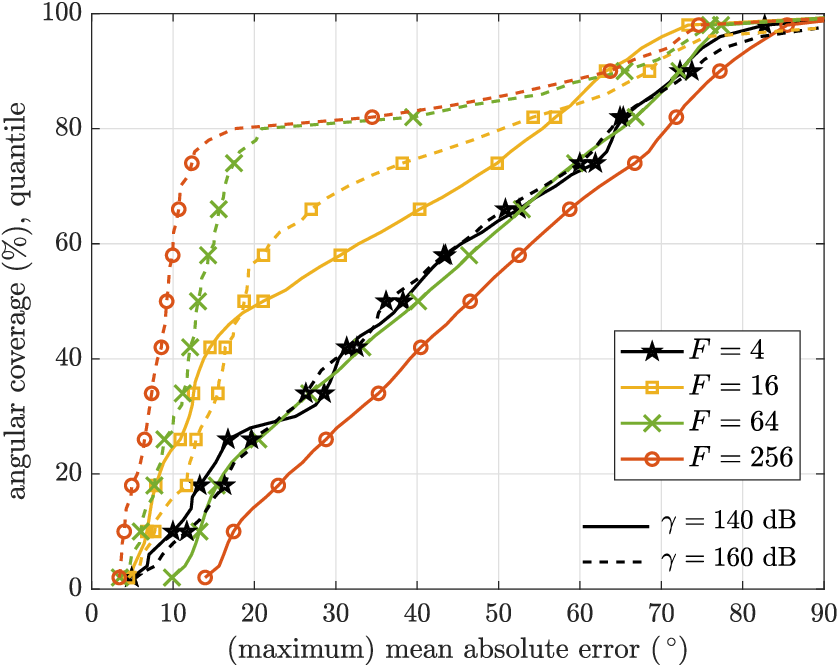}\\
    \includegraphics[width=0.9\linewidth]{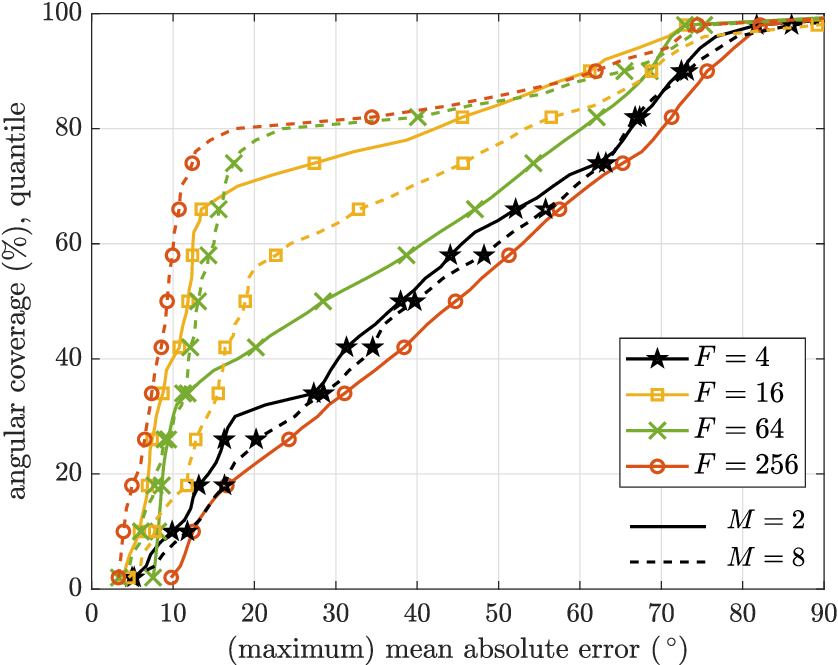}
    \caption{Angular coverage (or quantile) achieved by RPA as a function of a target (maximum) mean absolute estimation error for $\varphi$, for $F\in\{4,16,64,256\}$ and $\gamma\in\{140,160\}$ dB (top) and $M\in\{2,8\}$ (bottom). We assess this by considering the ground-truth $\varphi$ in the feasible estimation range \eqref{the3}.}
    \label{fig:FigFq}
\end{figure}

Fig.~\ref{fig:FigFq} depicts the angular coverage achieved by RPA with respect to its feasible angular estimation range for a given target (maximum) mean absolute error. This is shown for $F\in\{4,16,64,256\}$ and two different configurations of $\gamma$ and $M$. Note that as $F$ increases, the angular coverage may increase or decrease depending on the target accuracy, as expected from our previous discussions around Fig.~\ref{fig:FigFK}. Meanwhile, and as expected, a greater $\gamma$ improves the estimation accuracy, and thus angular coverage, given a target estimation accuracy. The gains become more noticeable and evident as $F$ increases. The reason is that the estimation error comes mainly from the signal noise as $F$ increases, which can be mitigated easily by strengthening the sensing signal power. This is not the case given a relatively small $F$, for which the estimation error is dominated by the quantization error induced by a discrete set of candidate angles under consideration. Indeed, we can observe in the figure that the maximum mean absolute estimation errors for a $60\%$ angular coverage of the feasible estimation region can be approximately reduced $0^\circ,10^\circ,33^\circ$, and $44^\circ$ for $F=4,16,64,$ and $256$, respectively, by increasing the transmit SNR from $140$ to $160$ dB. The impact of $M$ is trickier, as the estimation errors approximately increase $5^\circ$ and  $11^\circ$ for $F=4$ and $F=16$, respectively, and decrease $25^\circ$ and $43^\circ$ for $F=64$ and $F=256$, respectively,  by increasing $M$ from 2 to 8 for a $60\%$ angular coverage of the feasible estimation region. This is because beams become narrower as $M$ increases, which, although it leads to a higher SNR, may reinforce the quantization errors that dominate in the small-$F$ regime.
\begin{figure}[t!]
    \centering    \includegraphics[width=0.9\linewidth]{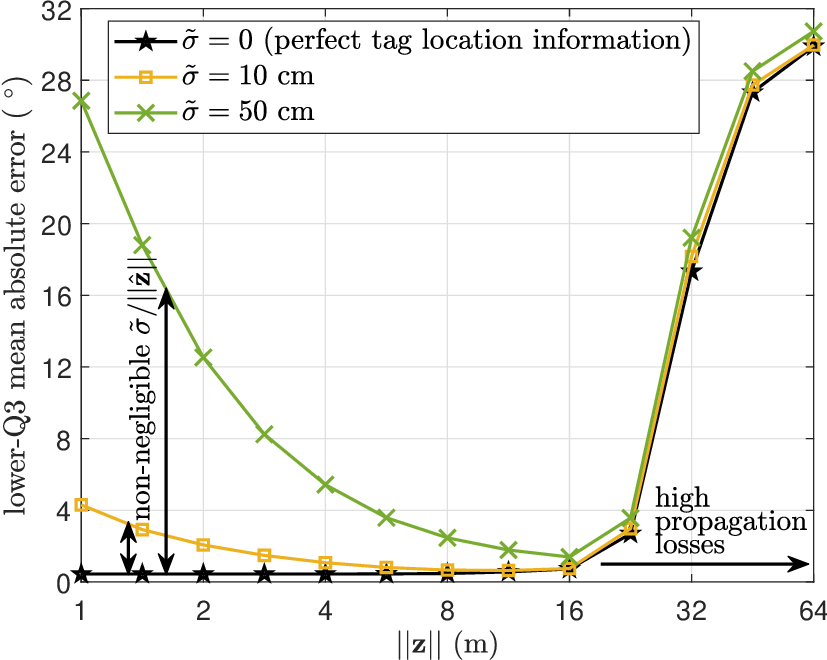}
    \caption{Lower-Q3 mean absolute estimation error of $\varphi$ as a function of the radar-target distance for RPA with $F=64$ and $\tilde{\sigma}\in\{0, 10,50\}$ cm. We assess this by considering the ground-truth $\varphi$ in the feasible estimation range \eqref{the3}.}
    \label{fig:FigR}    
\end{figure}
\subsection{Impact of the Sensing Distance}\label{ISD}
The impact of the sensing distance $||\mathbf{z}||$ on the estimation error within the feasible angular estimation range is illustrated in Fig.~\ref{fig:FigR} for $\tilde{\sigma}\in\{0, 10, 50\}$ cm. These results corroborate the insights and discussions in Sections~\ref{staS} and \ref{secImp}: the estimation error statistics, i.e., $\tilde{\sigma}$, have no significant performance impact for $||\mathbf{z}||\gg \tilde{\sigma}$. However, the detrimental impact of $\tilde{\sigma}$ becomes increasingly noticeable as $||\mathbf{z}||$ decreases since the operating assumptions from Section~\ref{staS} hold increasingly weaker. Meanwhile, as $||\mathbf{z}||$ increases, the performance degrades (similarly for any $\tilde{\sigma}$) due to the increasingly higher propagation losses. Fortunately, this effect can be mitigated, e.g., by increasing the transmit power, reducing $F$ such that each tone is allocated more power, or equipping more antennas at the radar for increased beamforming gains as discussed in Section~\ref{AngCov}.

\section{Conclusions} \label{conclusions}
In this work, we presented a novel RF sensing approach for accurate object orientation estimation using a \textit{dumb} LWA-equipped backscattering tag within a radar system. Our framework includes comprehensive signal modeling, geometrical constraints assessment, and the development of several orientation estimators of varying complexity. Specifically, we formulated an MLE for the tag's orientation and its simplified form given perfect tag location information, namely P$-$MLE. In addition, we analyzed the effects of imperfect tag location information, and with related results derived A$-$MLE and RPA. Surprisingly, the magnitude of the tag location estimation error does not significantly affect the orientation estimation accuracy as long as the sensing distance is relatively large. We derived the feasible orientation estimation region for RPA, and showed that it depends mainly on the system bandwidth and that it also applies approximately to the other estimators. Monte Carlo simulations corroborated our analytical insights and evinced that P$-$MLE is only appealing given perfect tag location information, which is never available in practice. Meanwhile, they revealed that A$-$MLE and RPA perform near-optimally and achieve high accuracy given imperfect tag location information, the latter with significantly low complexity. Moreover, we showed that potential multi-path effects may be negligible, there is an optimum number of tones that increases with the sensing time given a power budget,  performance gains from higher SNR increase with the number of tones, and operating with a large number of radar antennas may be discouraged for sensing signals comprising a small number of tones. Finally, note that this work can be expanded in several relevant directions, such as i) addressing the impact of clutter and interference from unintended scatterers in realistic environments; ii) exploring more complex sensing tasks, such as 3D or multi-tag orientation estimation and joint positioning and orientation; and iii) validating experimentally the theoretical findings in this paper using fabricated LWA-equipped tags and dedicated radar platforms while exploring hardware-related non-idealities.

\appendices
\section{LWA Example Radiation  Pattern}\label{appA}
Let $k_z(\lambda)=\beta(\lambda)-j\alpha(\lambda)$, wherein $\alpha$ and $\beta$ denote respectively the leakage rate and phase constant, and depend on the operation wavelength. For guided waves (such as those in microstrip lines or other types of transmission lines), for instance,  $\beta$ is determined by the effective refraction index $n_{\text{eff}}$ of the mode of propagation, i.e., $\beta(\lambda)\approx 2\pi n_{\text{eff}}/\lambda$. Note that $n_{\text{eff}}$ depends on the dielectric properties of the materials used in the construction of the transmission line or antenna structure and may vary slightly with frequency due to dispersion, i.e., the dependence of the dielectric constant on frequency. Meanwhile, $\alpha$ obeys
\begin{align}    \alpha(\lambda)\approx\frac{2\pi\sqrt{\epsilon_{\text{eff}}}}{\lambda}\times\frac{\text{perturbation factor}\ (\lambda)}{\text{structural integrity}\ (\lambda)}, \nonumber
\end{align}
where $\epsilon_{\text{eff}}$ is the effective permittivity of the structure, while the perturbation factor captures how the design modifications, e.g., slots, cuts, or periodic structures, influence the guided wave, and the structural integrity captures how effectively the structure confines/guides electromagnetic waves, maintaining their propagation along the intended path within the antenna.

Assuming a grating LWA with antenna length $l_1$ and width $l_2>\lambda$, as illustrated in Fig.~\ref{fig1}, one has that $n_\text{eff}\approx \sqrt{\epsilon_\text{eff}}$ \cite{Schwering.1983}. Also, for such a case, the far-field complex gain in the spatial direction determined by $\theta$ and $\phi$ is given by \cite{Schwering.1983} 
 \begin{align}     R(\theta,\!\phi;\!\lambda)\!=\!\! \sqrt{\!G(\lambda) S(\phi;\!\lambda)T(\theta,\phi;\!\lambda)}\!\exp\!\big(j\beta(\lambda) l_1 \cos\theta\big),\label{rad}
 \end{align}
 where the directivity gain of the antenna, $H$-plane pattern, $E-$plane pattern, and main-lobe pointing angle are respectively given by
 \begin{align}
      G(\lambda)&=\frac{64l_2}{\alpha(\lambda)\pi} \tanh\Big(\frac{\alpha(\lambda) l_1}{2}\Big)\cos\big(\theta_0(\lambda)\big),\\
S(\phi;\lambda)&=\! \sin^2\!\phi \cos^2\!\!\Big(\frac{\pi l_2}{\lambda} \cos \phi\Big) \Big(\!1\!-\!\frac{4l_2^2}{\lambda^2}\!\cos^2 \phi\Big)^{-2}\!\!\!,
\end{align}
\begin{align}
     T(\theta, \phi;\lambda)&=\left(\frac{\alpha(\lambda) l_1}{1-\exp(-\alpha(\lambda) l_1)}\right)^2\Bigg[ 1-2\exp(-\alpha(\lambda) l_1)\nonumber \\
     \times \cos\Big(&\frac{2\pi l_1}{\lambda}\big(\sin \theta \sin \phi - \sin \theta_0(\lambda)\big) \Big) +\exp(-2\alpha(\lambda) l_1) \Bigg]\nonumber\\
     \Bigg/ \Big(\alpha&(\lambda)^2 l_1^2+\frac{4\pi^2 l_1^2}{\lambda^2}(\sin\theta \sin\phi - \sin \theta_0(\lambda))^2 \Big),\\
     \theta_0(\lambda)&= \sin^{-1}\big(n_\text{eff}-\lambda/d' \big).\label{thetarad}
 \end{align}
 Herein, $d'$ is the grating period, which must be set to satisfy 
\begin{align}\label{deq}
   \frac{\lambda}{n_\text{eff}+1}\le d'\le \frac{\mu\lambda}{n_\text{eff}-1}, 
\end{align}
with $\mu\!=\!1$ if $n_\text{eff}\!>\!3$, otherwise $\mu\!=\!2$. Meanwhile, the half-power beamwidth and radiation efficiency respectively satisfy
\begin{align}
    \Theta(\lambda) &\approx \frac{\lambda}{l_1\cos\big( \theta_0(\lambda)\big)}, \label{Deltarad}\\
    \eta(\lambda) &= 1-\exp(-2\alpha(\lambda) l_1). \label{eff}
\end{align}

At the design (pre-manufacturing) phase, the antenna geometry (e.g., $l_1$, $l_2$, $d'$) and materials influencing $\alpha$ and $\beta$ are selected according to the desired frequency-dependent radiation pattern \cite{Oliner.2007}. Let's assume that the desired operation wavelength is in the range $[\lambda_\text{min}, \lambda_\text{max}]$. Then, from \eqref{deq}, one can select $d'$ such that
\begin{align}
    \frac{\lambda_\text{max}}{n_\text{eff}+1}\le d'\le \frac{2^{(\mathrm{sgn}(3-n_{\text{eff}})+1)/2}\lambda_\text{min}}{n_\text{eff}-1},\label{dineq}
\end{align}
while simultaneously considering the desired $\theta_0(\lambda)$, $\forall \lambda\in [\lambda_\text{min}, \lambda_\text{max}]$, according to \eqref{thetarad}. Then, one can set a target radiation efficiency $\eta^\circ$, e.g., $90\%$, and then make the design according to \eqref{Deltarad} and such that
\begin{align}
    \alpha(\lambda) l_1\ge -\ln(1-\eta^\circ)/2,\ \forall \lambda\in[\lambda_\text{min}, \lambda_\text{max}],
\end{align}
which comes from \eqref{eff}.

Fig.~\ref{fig:rad} illustrates the radiation pattern, both in terms of amplitude gain and phase response, for a given LWA implementation using the framework described herein. Specifically, we adopt $l_1=5$ cm, $l_2=1$ cm, and target an operation frequency range of $34-54$ GHz, corresponding to $\lambda\in [5.6, 8.8]$ mm. This allows using $\epsilon_\text{eff}=12$ and $n_\text{eff}=\sqrt{\epsilon_\text{eff}}$ \cite{Schwering.1983}. Moreover, we adopt $d'=2.1$ mm, which lies in the middle of the feasible range characterized by \eqref{dineq}. Finally, we use 
\begin{align}  
\alpha(\lambda)\lambda=-\frac{0.36}{(\lambda_\text{min}-\lambda_\text{max})^2}\Big(\lambda-\frac{\lambda_\text{min}+\lambda_\text{max}}{2}\Big)^2+0.1
\end{align}
to model the concave-quadratic behavior of $\alpha(\lambda)\lambda$ around the central operation wavelength as illustrated in \cite[Fig.~11]{Schwering.1983}.
\section{Proof of Lemma~\ref{lemmaH}} \label{lemH}
Using \eqref{hieq}, we have that
\begin{align}\label{h_i}
    h_m(\hat{\mathbf{z}}\!-\!\Delta \mathbf{z}) \!=\! \frac{\lambda}{4\pi||\hat{\mathbf{z}}_m'\!-\!\Delta \mathbf{z}||}\exp\Big(\!\!-\!\frac{2\pi}{\lambda} j ||\hat{\mathbf{z}}_m'\!-\!\Delta \mathbf{z}||\Big),
\end{align}
where $\hat{\mathbf{z}}_m'\triangleq \mathbf{x}_m-\hat{\mathbf{z}}$. Now, note that $||\hat{\mathbf{z}}_m'||\le ||\hat{\mathbf{z}}_m'-\Delta \mathbf{z}||\le ||\hat{\mathbf{z}}_m'|| + ||\Delta \mathbf{z}||$, and the bounds are tight for $||\Delta\mathbf{z}||\ll ||\hat{\mathbf{z}}_m'||$, which is the case of practical interest. Indeed, we can safely ignore the effect of $\Delta\mathbf{z}$ on the path loss coefficient, but not that on the signal phase shift unless $||\Delta \mathbf{z}||\ll \lambda$ is guaranteed, which might only happen when using ultra-accurate positioning techniques, e.g., based on carrier (tone) phase measurements \cite{Talvitie.2023,Lopez.2023}. Hence, herein, we adopt the lower-bound of $||\hat{\mathbf{z}}_m'\!-\!\Delta\mathbf{z}||$ for the path loss component, and its upper-bound, capturing the impact of $\Delta\mathbf{z}$, for the phase shift component. Then, \eqref{h_i} can be approximated as
\begin{align}
    h_m(\hat{\mathbf{z}}\!-\!\Delta \mathbf{z}) &\approx \frac{\lambda}{4\pi||\hat{\mathbf{z}}_m'||}\exp\big(-2\pi j\big( ||\hat{\mathbf{z}}_m'||\!+\!||\Delta \mathbf{z}||\big)/\lambda\big)\nonumber\\
    &=h_m(\hat{\mathbf{z}})\exp(-2\pi j ||\Delta\mathbf{z}||/\lambda),\label{h_ia}
\end{align}
thus, $\mathbf{h}_i(\hat{\mathbf{z}}-\Delta\mathbf{z})\approx \mathbf{h}_i(\hat{\mathbf{z}}) \exp(-2\pi j||\Delta\mathbf{z}||/\lambda_i)$. Using this and \eqref{wi}, we obtain
\begin{align}
    \mathbf{w}_i^H \mathbf{h}_i(\hat{\mathbf{z}}-\Delta \mathbf{z})&\approx \frac{\mathbf{h}_i(\hat{\mathbf{z}})^H}{||\mathbf{h}_i(\hat{\mathbf{z}})||}\mathbf{h}_i(\hat{\mathbf{z}})\exp(-2\pi j ||\Delta\mathbf{z}||/\lambda_i)\nonumber\\
    &= ||\mathbf{h}_i(\hat{\mathbf{z}})||\exp(-2\pi j ||\Delta\mathbf{z}||/\lambda_i).
\end{align}
Now, since $||\Delta\mathbf{z}||$ has a continuous and unbounded distribution (as long as $\Tr(\mathbf{\Sigma})\ne 0$), we have that 
$\mathbf{w}_i^H \mathbf{h}_i(\hat{\mathbf{z}}-\Delta \mathbf{z})\sim ||\mathbf{h}_i(\hat{\mathbf{z}})||\exp(j V)$
with $V$ is uniformly random in $[0,2\pi]$ due to the wrapping effect of the phase about $2\pi$. As per our previous discussions, this holds tight as $\Tr(\mathbf{\Sigma})/||\hat{\mathbf{z}}||^2\rightarrow 0$. \hfill\qedsymbol
\section{Proof of Lemma~\ref{lemmaPS}} \label{lemPS}
We have $\psi(\hat{\mathbf{z}}-\Delta\mathbf{z})=\tan^{-1}\big(\frac{\hat{z}_2-\Delta z_2}{\hat{z}_1-\Delta z_1}\big)$, where $\Delta z_1\sim \sigma_x X$, $\Delta z_2\sim \sigma_y Y$, and $X, Y\sim \mathcal{N}(0,1)$. Let us put this with some notation abuse as $\psi(X,Y) = \tan^{-1} \big(\frac{\hat{z}_2-\sigma_y Y}{\hat{z}_1-\sigma_x X}\big)$. Since $\hat{z}_2,\hat{z}_1\gg \sigma_y,\sigma_x$, a good approximation for this comes from its linear Taylor series approximation around $\mathbb{E}[X]=\mathbb{E}[Y]=0$, for which the impact of $\sigma_y,\sigma_x$ disappears.  That is 
\begin{align}
    \psi(X,Y)\!&\approx\! \psi(0,0)\! +\! \frac{\partial \psi(X,Y)}{\partial X}\bigg|_{(0,0)}\!\!\!\!\!X \!+\! \frac{\partial \psi(X,Y)}{\partial Y}\bigg|_{(0,0)}\!\!\!\!\!Y \nonumber\\    
    &= \tan^{-1}\Big(\frac{\hat{z}_2}{\hat{z}_1}\Big) + \frac{\sigma_x \hat{z}_2}{\hat{z}_1^2+\hat{z}_2^2} Y-\frac{\sigma_y \hat{z}_1}{\hat{z}_1^2+\hat{z}_2^2}X.
\end{align}
Then, using the distribution of $X$ and $Y$, we reach \eqref{PsiXY}. \hfill\qedsymbol
\section{Proof of Theorem~\ref{the0}}\label{THE0}
When tag location estimation
errors are much smaller than actual distance estimations, we can exploit the results and insights from Lemmas~\ref{lemmaH} and \ref{lemmaPS}. Specifically, let us substitute $\mathbf{w}_i^H\mathbf{h}_i(\hat{\mathbf{z}}-\Delta\mathbf{z})$ by $||\mathbf{h}_i(\hat{\mathbf{z}})||\exp(jV)$ and $\psi(\hat{\mathbf{z}}-\Delta\mathbf{z})$ by $\psi(\hat{\mathbf{z}})$ into \eqref{logL} to obtain
\begin{align}
     \ln \mathsf{L}(\varphi;\tilde{s}'')&\propto \sum_{i=1}^F\sum_{k=1}^K\ln \mathbb{E}_{V}\Big[\!\exp\Big(\!\!-\! |\delta_i'(\varphi)\exp(jV)|^2|s_i[k]|^2 \nonumber\\
    &\qquad\ \ \ +2\Re\big\{\tilde{s}_i''[k]^*\!s_i[k]\delta_i'(\varphi)\exp(jV)\big\}\Big)\Big]\nonumber\\
    &= \sum_{i=1}^F\sum_{k=1}^K \Big(\ln\mathbb{E}_{V}\Big[\exp\Big(2\Re\big\{\tilde{s}_i''[k]^*\!s_i[k]\delta_i'(\varphi)\nonumber\\
    &\qquad\ \ \ \times \exp(jV)\big\}\Big)\Big] - |\delta_i'(\varphi)|^2|s_i[k]|^2\Big),\label{lnL}
\end{align}
where the last line comes from exploiting $|a\exp(jV)|=|a|$, $\exp(a+b)=\exp(a)\exp(b)$, and $\ln\exp(a) = a$. Now, let $c=2\tilde{s}_i''[k]^*\!s_i[k]\delta_i'(\varphi)$, so to proceed further we need to compute $E_V=\mathbb{E}_V[\exp(\Re\{c\exp(jV)\})]$ as pursued in the following
\begin{align}
E_V&= \mathbb{E}_V\big[\exp\big(\Re\{(\Re\{c\}+j\Im\{c\})(\cos V+j\sin V)\}\big)\big]\nonumber\\
&\stackrel{(a)}{=}\mathbb{E}_V\big[\exp\big(\Re\{c\}\cos V - \Im\{c\}\sin V\big)\big]\nonumber\\
&\stackrel{(b)}{=}\!\frac{1}{2\pi}\!\int_{0}^{2\pi}\!\!\!\!\!\exp\big(\Re\{c\}\!\cos v\!-\!\Im\{c\}\!\sin v\big)\mathrm{d}v
\!\stackrel{(c)}{=}\! I_0(|c|),\label{Ic}
\end{align}
where $(a)$ comes from simple algebraic simplifications, $(b)$ from using the probability density function of $V$, given by $1/(2\pi), \forall v\in[0,2\pi]$, to state the expectation in integral form, and $(c)$ from exploiting \cite[eq.(3.338.4)]{Gradshteyn.2014}. Finally, substituting \eqref{Ic} into \eqref{lnL} followed by some simple algebraic transformations, we obtain \eqref{MLEA}. \hfill\qedsymbol
\section{Proof of Theorem~\ref{lem1}}\label{appB}
Using the factorization theorem \cite{Kay.1993}, $u_i^*$ is a sufficient statistic for 
$\delta_i'(\varphi)$, and thus for $|\delta_i'(\varphi)|$. Note that $u_i^*\!\sim\! \mathcal{CN}\big(\delta_i'(\varphi), 1/\sum_{k=1}^K|s_i[k]|^2\big)$, thus $|u_i|$ follows a Rice distribution with parameters $|\delta_i'(\varphi)|$ and $\big(2\sum_{k=1}^K|s_i[k]|^2\big)^{-1/2}$\!\!, and
\begin{align}
    \mathbb{E}[|u_i|]\!=\!\frac{1}{2}\sqrt{\!\frac{\pi}{\sum_{k=1}^K\!\!|s_i[k]|^2}}\!\!\ _1F_1\!\Big(\!\!-\!\frac{1}{2},1,\!-|\delta_i'(\varphi)|^2\sum_{k=1}^K\!|s_i[k]|^2\Big).\label{Eu}
\end{align}
Then, by substituting $\mathbb{E}[|u_i|]$ by $|u_i|$ in \eqref{Eu}, thus, guaranteeing unbiased estimation, and performing some simple algebraic transformations, we obtain \eqref{F11}.

For asymptotically low and high SNR scenarios, i.e., $|\delta_i'(\varphi)|\rightarrow 0$ and $|\delta_i'(\varphi)|\rightarrow \infty$, we can exploit results from \cite{Olver.2010} to approximate $ _1F_1(-1/2,1,x)\rightarrow 1- x/2\ \ \text{as}\ \  x\rightarrow 0,$ and $ _1F_1(-1/2,1,x)\rightarrow 2\sqrt{-\frac{x}{\pi}}\ \ \text{as}\ \ x\rightarrow -\infty$. Substituting these into \eqref{F11}, we obtain \eqref{lowS} and \eqref{highS}. \hfill\qedsymbol
\bibliographystyle{IEEEtran}
\bibliography{IEEEabrv,references}

\begin{thebibliography}{10}
\providecommand{\url}[1]{#1}
\csname url@samestyle\endcsname
\providecommand{\newblock}{\relax}
\providecommand{\bibinfo}[2]{#2}
\providecommand{\BIBentrySTDinterwordspacing}{\spaceskip=0pt\relax}
\providecommand{\BIBentryALTinterwordstretchfactor}{4}
\providecommand{\BIBentryALTinterwordspacing}{\spaceskip=\fontdimen2\font plus
\BIBentryALTinterwordstretchfactor\fontdimen3\font minus \fontdimen4\font\relax}
\providecommand{\BIBforeignlanguage}[2]{{%
\expandafter\ifx\csname l@#1\endcsname\relax
\typeout{** WARNING: IEEEtran.bst: No hyphenation pattern has been}%
\typeout{** loaded for the language `#1'. Using the pattern for}%
\typeout{** the default language instead.}%
\else
\language=\csname l@#1\endcsname
\fi
#2}}
\providecommand{\BIBdecl}{\relax}
\BIBdecl

\bibitem{Suresh.2021}
N.~Kumar \emph{et~al.}, ``A review on metamaterials for device applications,'' \emph{Crystals}, vol.~11, no.~5, p. 518, May 2021.

\bibitem{Padilla.2022}
W.~Padilla and R.~Averitt, ``Imaging with metamaterials,'' \emph{Nat. Rev. Phys.}, vol.~4, no.~2, pp. 85--100, Feb. 2022.

\bibitem{Xue.2023}
H.~Xue \emph{et~al.}, ``Radiofrequency sensing systems based on emerging two-dimensional materials and devices,'' \emph{Int. J. Extreme Manuf.}, vol.~5, no.~3, p. 032010, Jul. 2023.

\bibitem{Shirehjini.2012}
A.~Shirehjini \emph{et~al.}, ``An {RFID}-based position and orientation measurement system for mobile objects in intelligent environments,'' \emph{IEEE Trans. Instrum. Meas.}, vol.~61, no.~6, pp. 1664--1675, Jun. 2012.

\bibitem{Wei.2016}
T.~Wei and X.~Zhang, ``Gyro in the air: tracking {3D} orientation of bat- teryless {Internet-of-Things},'' in \emph{Proc. MobiCom}, Oct. 2016, pp. 55--68.

\bibitem{Figueiredo.2023}
G.~Figueiredo \emph{et~al.}, ``Monitoring head orientation using passive {RFID} tags,'' \emph{IEEE J. Radio Freq. Identif.}, vol.~7, pp. 582--590, Oct. 2023.

\bibitem{Krigslund.2012}
R.~Krigslund \emph{et~al.}, ``Orientation sensing using multiple passive {RFID} tags,'' \emph{IEEE Antennas Wirel. Propag. Lett.}, vol.~11, pp. 176--179, Jan. 2012.

\bibitem{Genovesi.2018}
S.~Genovesi \emph{et~al.}, ``Chipless radio frequency identification\! {(RFID)} sensor for angular rotation monitoring,'' \emph{Technologies}, vol.~6, no.~3, p.~61, Jun. 2018.

\bibitem{Barbot.2020}
N.~Barbot \emph{et~al.}, ``Angle sensor based on chipless {RFID} tag,'' \emph{IEEE Antennas Wirel. Propag. Lett.}, vol.~19, no.~2, pp. 233--237, Nov. 2020.

\bibitem{Liu.2024}
H.~Liu \emph{et~al.}, ``Simultaneous detection of the orientation and position of moving objects with simple {RFID} array for industrial {IoT} applications,'' \emph{IEEE Internet Things J.}, pp. 1--1, 2024.

\bibitem{Chang.2020}
K.~Chang \emph{et~al.}, ````{Wireless Paint}'': Code design for {3D} orientation estimation with backscatter arrays,'' in \emph{IEEE ISIT}, Jun. 2020, pp. 1224--1229.

\bibitem{Rammal.2023}
M.~Rammal \emph{et~al.}, ``Coded estimation: Design of backscatter array codes for {3D} orientation estimation,'' \emph{IEEE Trans. Wirel. Commun.}, vol.~22, no.~9, pp. 5844--5854, Jan. 2023.

\bibitem{Srinivasarao.1999}
M.~Srinivasarao, ``"{Nano-Optics in the Biological World: Beetles, Butterflies, Birds, and Moths}",'' \emph{Chemical Reviews}, vol.~99, no.~7, 1999.

\bibitem{Kludze.2022}
A.~Kludze, R.~Shrestha, C.~Miftah, E.~Knightly, D.~Mittleman, and Y.~Ghasempour, ``Quasi-optical {3D} localization using asymmetric signatures above 100 {GHz},'' in \emph{Proceedings of the 28th Annual International Conference on Mobile Computing And Networking}, ser. MobiCom '22.\hskip 1em plus 0.5em minus 0.4em\relax New York, NY, USA: Association for Computing Machinery, 2022, pp. 120--132.

\bibitem{Poveda.2021}
M.~Poveda-Garc{\'\i}a \emph{et~al.}, ``Frequency-scanned leaky-wave antenna topologies for two-dimensional direction of arrival estimation in {IoT} wireless networks,'' in \emph{EuCAP}, Mar. 2021, pp. 1--5.

\bibitem{Martinez.2022}
A.~Gil-Mart{\'\i}nez \emph{et~al.}, ``Direction finding of {RFID} tags in {UHF} band using a passive beam-scanning leaky-wave antenna,'' \emph{IEEE J. Radio Freq. Identif.}, vol.~6, pp. 552--563, Jun. 2022.

\bibitem{MartinezPoveda.2022}
------, ``Frequency-scanned monopulse antenna for {RSSI}-based direction finding of {UHF RFID} tags,'' \emph{IEEE Antennas Wirel. Propag. Lett.}, vol.~21, no.~1, pp. 158--162, Jan. 2022.

\bibitem{Neophytou.2022}
K.~Neophytou \emph{et~al.}, ``Simultaneous user localization and identification using leaky-wave antennas and backscattering communications,'' \emph{IEEE Access}, vol.~10, pp. 37\,097--37\,108, Mar. 2022.

\bibitem{Martinez.2023}
A.~Gil-Martinez \emph{et~al.}, ``Monopulse leaky wave antennas for {RSSI}-based direction finding in wireless local area networks,'' \emph{IEEE Trans. Antennas Propag.}, vol.~71, no.~11, pp. 8602--8615, Nov. 2023.

\bibitem{Haofan.2023}
\BIBentryALTinterwordspacing
H.~Lu, M.~Mazaheri, R.~Rezvani, and O.~Abari, ``A millimeter wave backscatter network for two-way communication and localization,'' in \emph{Proceedings of the ACM SIGCOMM 2023 Conference}, ser. ACM SIGCOMM '23.\hskip 1em plus 0.5em minus 0.4em\relax New York, NY, USA: Association for Computing Machinery, 2023, pp. 49--61. [Online]. Available: \url{https://doi.org/10.1145/3603269.3604873}
\BIBentrySTDinterwordspacing

\bibitem{Morais.2020}
C.~de~Lima \emph{et~al.}, ``{6G White Paper} on localization and sensing,'' \emph{6G Research Visions}, vol.~12, 2020.

\bibitem{Zhao.2021}
M.~Zhao \emph{et~al.}, ``{ULoc:} low-power, scalable and cm-accurate {UWB-Tag} localization and tracking for indoor applications,'' \emph{Proc. ACM IMWUT}, vol.~5, no.~3, Sep. 2021.

\bibitem{Soltanaghaei.2021}
E.~Soltanaghaei \emph{et~al.}, ``Millimetro: {mmWave} retro-reflective tags for accurate, long range localization,'' in \emph{Proc. MobiCom}.\hskip 1em plus 0.5em minus 0.4em\relax New York, NY, USA: ACM, Oct. 2021, pp. 69--82.

\bibitem{Goodman.2018}
N.~A. Goodman, ``Foundations of cognitive radar for next-generation radar systems,'' in \emph{Academic Press Library in Signal Processing, Volume 7}.\hskip 1em plus 0.5em minus 0.4em\relax Elsevier, 2018, pp. 153--195.

\bibitem{Nunez.2023}
J.~M. N{\'u}{\~n}ez-Ortu{\~n}o, J.~P. Gonz{\'a}lez-Coma, R.~Nocelo~L{\'o}pez, F.~Troncoso-Pastoriza, and M.~{\'A}lvarez-Hern{\'a}ndez, ``Beamforming techniques for passive radar: An overview,'' \emph{Sensors}, vol.~23, no.~7, p. 3435, 2023.

\bibitem{Komatsu.2021}
K.~Komatsu \emph{et~al.}, ``Theoretical analysis of in-band full-duplex radios with parallel {Hammerstein} self-interference cancellers,'' \emph{IEEE Trans. Wirel. Commun.}, vol.~20, no.~10, pp. 6772--6786, Oct. 2021.

\bibitem{He.2022}
Y.~He \emph{et~al.}, ``Frequency-domain successive cancellation of nonlinear self-interference with reduced complexity for full-duplex radios,'' \emph{IEEE Trans. Commun.}, vol.~70, no.~4, pp. 2678--2690, Apr. 2022.

\bibitem{Abdelghaffar.2024}
M.~Abdelghaffar, T.~V.~P. Santhappan, Y.~Tokgoz, K.~Mukkavilli, and a.~Tingfang~Ji, ``Subband full-duplex large-scale deployed network designs and tradeoffs,'' \emph{Proceedings of the IEEE}, vol. 112, no.~5, pp. 487--510, 2024.

\bibitem{Ozkaptan.2018}
C.~Ozkaptan \emph{et~al.}, ``{OFDM} pilot-based radar for joint vehicular communication and radar systems,'' in \emph{IEEE VNC}, 2018, pp. 1--8.

\bibitem{Kay.1993}
S.~M. Kay, \emph{Fundamentals of statistical signal processing: estimation theory, Volume i}.\hskip 1em plus 0.5em minus 0.4em\relax Prentice-Hall, Inc., 1993.

\bibitem{Welp.2020}
B.~Welp, G.~Briese, and N.~Pohl, ``Ultra-wideband {FMCW} radar with over 40 {GHz} bandwidth below 60 {GHz} for high spatial resolution in {SiGe BiCMOS},'' in \emph{IEEE/MTT-S International Microwave Symposium (IMS)}, 2020, pp. 1255--1258.

\bibitem{Lopez.2023}
O.~L\'opez \emph{et~al.}, ``Polarization diversity-enabled {LOS/NLOS} identification via carrier phase measurements,'' \emph{IEEE Trans. Commun.}, vol.~71, no.~3, pp. 1678--1690, Mar. 2023.

\bibitem{Schwering.1983}
F.~Schwering and S.~Peng, ``Design of dielectric grating antennas for millimeter-wave applications,'' \emph{IEEE Trans. Microw. Theory Tech.}, vol.~31, no.~2, pp. 199--209, Feb. 1983.

\bibitem{Oliner.2007}
A.~Oliner \emph{et~al.}, ``Leaky-wave antennas,'' \emph{Antenna engineering handbook}, vol.~4, p.~12, 2007.

\bibitem{Talvitie.2023}
J.~Talvitie, M.~S{\"a}ily, and M.~Valkama, ``Orientation and location tracking of {XR} devices: {5G} carrier phase-based methods,'' \emph{IEEE Journal of Selected Topics in Signal Processing}, vol.~17, no.~5, pp. 919--934, 2023.

\bibitem{Gradshteyn.2014}
I.~Gradshteyn and I.~Ryzhik, \emph{Table of integrals, series, and products}.\hskip 1em plus 0.5em minus 0.4em\relax Academic press, 2014.

\bibitem{Olver.2010}
F.~Olver \emph{et~al.}, \emph{{NIST} handbook of mathematical functions hardback and {CD-ROM}}.\hskip 1em plus 0.5em minus 0.4em\relax Cambridge university press, 2010.

\end{thebibliography}
\end{document}